\documentclass[twoside,leqno]{article}

\usepackage{amsmath, amssymb}
\usepackage{balance}
\usepackage{cite}
\usepackage{enumitem}
\usepackage[letterpaper]{geometry}
\usepackage{ltexpprt}
\usepackage{hyperref}
\usepackage{xcolor}

\usepackage{booktabs}
\usepackage{multirow}
\usepackage{cite}
\usepackage{graphicx}
\usepackage{subcaption}

\usepackage{algorithm}
\usepackage[noend]{algorithmic}

\newtheorem{observation}{Fact}

\DeclareMathOperator*{\argmax}{arg\,max}

\begin{document}

\title{\Large Max-Min Diversification with Fairness Constraints: \\ Exact and Approximation Algorithms}
\author{
Yanhao Wang\thanks{East China Normal University. yhwang@dase.ecnu.edu.cn}
\and Michael Mathioudakis\thanks{University of Helsinki. michael.mathioudakis@helsinki.fi}
\and Jia Li\thanks{East China Normal University. jiali@stu.ecnu.edu.cn}
\and Francesco Fabbri\thanks{Spotify. francescof@spotify.com}
}

\date{}

\maketitle

\begin{abstract}\small\baselineskip=9pt Diversity maximization aims to select a diverse and representative subset of items from a large dataset. It is a fundamental optimization task that finds applications in data summarization, feature selection, web search, recommender systems, and elsewhere. However, in a setting where data items are associated with different groups according to sensitive attributes like sex or race, it is possible that algorithmic solutions for this task, if left unchecked, will under- or over-represent some of the groups. Therefore, we are motivated to address the problem of \emph{max-min diversification with fairness constraints}, aiming to select $k$ items to maximize the minimum distance between any pair of selected items while ensuring that the number of items selected from each group falls within predefined lower and upper bounds. In this work, we propose an exact algorithm based on integer linear programming that is suitable for small datasets as well as a $\frac{1-\varepsilon}{5}$-approximation algorithm for any $\varepsilon \in (0, 1)$ that scales to large datasets. Extensive experiments on real-world datasets demonstrate the superior performance of our proposed algorithms over existing ones.

\end{abstract}

\section{Introduction}
\label{sec:intro}

In recent years, algorithms have been increasingly used for data-driven automated decision-making in many domains of everyday life.
This has raised concerns about the possibility that algorithms may produce unfair and discriminatory decisions for specific population groups, particularly in sensitive socio-computational domains such as voting, hiring, banking, education, and criminal justice~\cite{DBLP:journals/fdata/Olteanu00K19, DBLP:journals/cacm/ChouldechovaR20}.
To alleviate such concerns, there has been a lot of research devoted to incorporating fairness into the algorithms for automated decision tasks, including classification~\cite{DBLP:conf/innovations/DworkHPRZ12}, clustering~\cite{DBLP:conf/nips/Chierichetti0LV17}, ranking~\cite{DBLP:conf/cikm/ZehlikeB0HMB17,DBLP:conf/aaai/NarasimhanCGW20}, matching~\cite{DBLP:conf/ijcai/SankarLNN21}, and data summarization~\cite{DBLP:conf/icml/CelisKS0KV18,DBLP:conf/icml/KleindessnerAM19}.

This paper considers the diversity maximization problem and addresses its fairness-aware variant.
The problem consists in selecting a diverse subset of items from a given dataset and is encountered in data summarization~\cite{DBLP:conf/icml/CelisKS0KV18,DBLP:conf/icdt/Moumoulidou0M21}, web search~\cite{DBLP:conf/wsdm/AgrawalGHI09}, recommendation~\cite{DBLP:journals/kbs/KunaverP17}, feature selection~\cite{DBLP:conf/aaai/ZadehGMZ17}, and elsewhere~\cite{DBLP:conf/cvpr/ZhangLPCS20}.
Existing literature on the problem of diversity maximization primarily focuses on two objectives, namely \emph{max-min diversification} (MMD), which aims to maximize the minimum distance between any pair of selected items, and \emph{max-sum diversification} (MSD), which seeks to maximize the sum of pairwise distances between selected items.
As shown in Figure~\ref{fig:example}, MMD tends to cover the data range uniformly, while MSD tends to pick ``outliers'' and may include highly similar items in the solution.
Since the notion of diversity captured by MMD better represents the property that data summarization, feature selection, and many other tasks target with their solutions, we will only consider MMD in this paper.
To be precise, given a set $V$ of $n$ items in a metric space and a positive integer $k \leq n$, MMD asks for a size-$k$ subset $S$ of $V$ to maximize the minimum pairwise distance within $S$.

In particular, we study the \emph{fair max-min diversification} (FMMD) problem, a variant of MMD that aims not only to maximize the diversity measure defined above but also to guarantee the satisfaction of group fairness constraints as described below.
Let all the items in $V$ be divided into $C$ disjoint groups $V_1, \ldots, V_C$ by a sensitive attribute such as sex or race.
To ensure a fair representation, the number of items selected from each group $V_c$, where $c \in [1, \ldots, C]$, is limited to be between lower and upper bounds specified as input.
This definition of group fairness constraints captures and generalizes several existing notions of \emph{fairness} for groups, including proportional representation~\cite{DBLP:conf/ijcai/CelisHV18,DBLP:conf/nips/HalabiMNTT20}, equal representation~\cite{DBLP:conf/icml/KleindessnerAM19,DBLP:conf/icml/JonesNN20}, and statistical parity~\cite{DBLP:conf/innovations/DworkHPRZ12,DBLP:conf/icml/ZemelWSPD13}, and has been widely used in optimization problems such as top-$k$ ranking~\cite{DBLP:conf/icalp/CelisSV18}, submodular maximization~\cite{DBLP:conf/nips/HalabiMNTT20}, and multiwinner voting~\cite{DBLP:conf/ijcai/CelisHV18}.

\begin{figure}
  \centering
  \subcaptionbox{MMD\label{fig:examplea}}[.4\linewidth]{\includegraphics[height=1.5in]{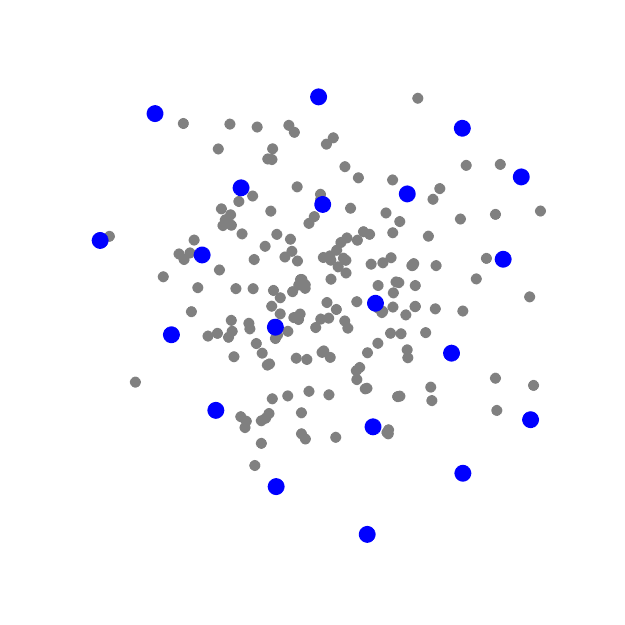}}
  \hspace{1em}
  \subcaptionbox{MSD\label{fig:exampleb}}[.4\linewidth]{\includegraphics[height=1.5in]{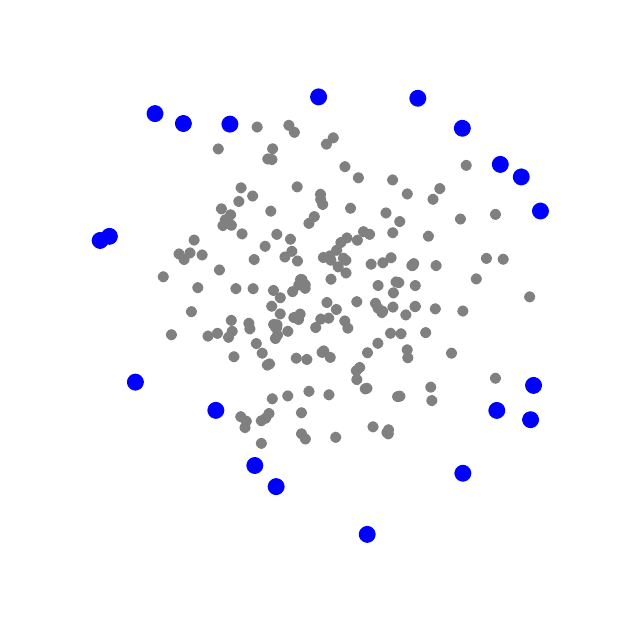}}
  \caption{An illustration of max-min diversification (MMD) vs.~max-sum diversification (MSD). The items in the solutions are marked in blue color.}\label{fig:example}
\end{figure}

\subsection{Related Work}
\label{subsec:literature}

Erkut~\cite{ERKUT199048} proved that the MMD problem is NP-hard in metric spaces.
Ravi \emph{et al.}~\cite{DBLP:journals/ior/RaviRT94} proposed a $\frac{1}{2}$-approximation greedy algorithm~\cite{DBLP:journals/tcs/Gonzalez85} for MMD and proved that no polynomial algorithm could achieve a better approximation factor unless P=NP.
Recently, many different algorithms have been proposed for MMD in various settings.
Indyk \emph{et al.}~\cite{DBLP:conf/pods/IndykMMM14} proposed a $\frac{1}{3}$-approximation distributed algorithm for MMD based on the notion of \emph{coresets}.
Drosou and Pitoura~\cite{DBLP:journals/tkde/DrosouP14} designed a $\frac{b-1}{2b^2}$-approximation cover tree-based algorithm for MMD on dynamic data, where $b$ is the base of the cover tree.
Ceccarello \emph{et al.}~\cite{DBLP:journals/pvldb/CeccarelloPPU17} proposed $(\frac{1}{2}-\varepsilon)$-approximate MapReduce and streaming algorithms for MMD in metric spaces of bounded doubling dimension.
Borassi \emph{et al.}~\cite{DBLP:conf/pods/BorassiELVZ19} proposed a sliding-window algorithm for MMD.
Nevertheless, none of the above algorithms are applicable to FMMD because they cannot guarantee the fulfillment of fairness constraints. 

Moumoulidou \emph{et al.}~\cite{DBLP:conf/icdt/Moumoulidou0M21} first proposed approximation algorithms for the fair variant of MMD.
Addanki \emph{et al.}~\cite{DBLP:conf/icdt/Addanki0MM22} improved the approximation ratios of the algorithms in~\cite{DBLP:conf/icdt/Moumoulidou0M21}.
Wang \emph{et al.}~\cite{DBLP:conf/icde/WangFM22} proposed two streaming algorithms for the fair variant of MMD.
However, these algorithms work for exact-size group fairness constraints, a special case of our bounded-size group fairness constraints.
Moreover, as shown empirically, these algorithms provide lower-quality solutions than ours.

Besides MMD, many other optimization problems have similar group fairness-aware variants -- e.g., determinantal point processes~\cite{DBLP:conf/icml/CelisKS0KV18}, $k$-centers~\cite{DBLP:conf/icml/KleindessnerAM19,DBLP:conf/icml/ChiplunkarKR20,DBLP:conf/icml/JonesNN20}, top-$k$ ranking~\cite{DBLP:conf/icalp/CelisSV18}, submodular maximization~\cite{DBLP:conf/nips/HalabiMNTT20,DBLP:conf/www/0001FM21}, and multiwinner voting~\cite{DBLP:conf/ijcai/CelisHV18}.
However, since their objectives differ from MMD, the algorithms proposed for their fair variants are not directly applicable to FMMD.

\subsection{Our Results}
\label{subsec:results}

The main results of this paper are two novel algorithms for the \emph{fair max-min diversification} (FMMD) problem, which selects a size-$k$ subset $S$ from a dataset $V$ that maximizes the diversity value while satisfying group-fairness constraints.

We first propose \textsf{FMMD-E}, an exact algorithm that is suitable for solving FMMD on small datasets, despite the NP-hardness of the problem. This algorithm exploits the connection between the MMD and maximum independent set (MIS) problems. It formulates FMMD as the problem of finding an independent set of vertices with group fairness constraints on an undirected graph. Then, the optimal solution of FMMD can be obtained in $O(n^k \log n)$ time by solving the reduced problem via integer-linear programming (ILP).

Since \textsf{FMMD-E} cannot scale to large datasets, we propose \textsf{FMMD-S}, a more scalable approximation algorithm for FMMD.
Specifically, for any $\varepsilon \in (0,1)$, \textsf{FMMD-S} provides $\frac{1-\varepsilon}{5}$-approximate solutions for FMMD in $O \big( C k n + C^k \log{\frac{1}{\varepsilon}} \big)$ time.
Under the assumptions that $k = o(\log{n})$ and $C = O(1)$, the time complexity of \textsf{FMMD-S} is reduced to $O\big( n(k + \log{\frac{1}{\varepsilon}}) \big)$.
The basic idea of \textsf{FMMD-S} is, in the first step, to limit the computation to a considerably smaller subset of the original dataset (i.e., \emph{coreset}) by running the greedy algorithm~\cite{DBLP:journals/tcs/Gonzalez85, DBLP:journals/ior/RaviRT94} and, in the second step, to use an ILP-based method similar to \textsf{FMMD-E} to obtain an approximate solution to FMMD from the subset.

Finally, we compare the performance of our algorithms with the state-of-the-art algorithms in~\cite{DBLP:conf/icdt/Moumoulidou0M21,DBLP:conf/icdt/Addanki0MM22,DBLP:conf/icde/WangFM22} for the FMMD problem on real-world datasets. The results show that \emph{i)} \textsf{FMMD-E} provides exact solutions in reasonable time on small datasets (e.g., $n = 1,000$); \emph{ii)} \textsf{FMMD-S} returns solutions of higher quality than existing approximation algorithms in comparable time while scaling to large datasets with millions of items.

\section{Preliminaries}
\label{sec:def}

In this section, we first formally define the FMMD problem, a fairness-aware variant of max-min diversification (MMD), and then provide its hardness result.

\paragraph{Max-Min Diversification (MMD).}
Let $V$ be a set of $n$ items and $d: V \times V \rightarrow \mathbb{R}_{\geq 0}$ be a distance metric that captures the dissimilarities between items in $V$.
We remind that, by definition, $d(\cdot,\cdot)$ satisfies the following properties for any $u, v, w \in V$: \emph{i)} $d(u, v) = 0 \Leftrightarrow u = v$ (identity of indiscernibles); \emph{ii)} $d(u, v) = d(v, u)$ (symmetry); \emph{iii)} $d(u, v) + d(v, w) \geq d(u, w)$ (triangle inequality). 
For MMD, the diversity value $div(S)$ of a subset $S \subseteq V$ is defined as the minimum among all pairwise distances between distinct items in $S$ -- i.e., $div(S) = \min_{u, v \in S \, : \, u \neq v} d(u, v)$. Given a set $V$, a distance function $d(\cdot, \cdot)$, and a positive integer $k \leq n$, the MMD problem asks for a size-$k$ subset $S$ of $V$ such that $div(S)$ is maximized.

\paragraph{Fair Max-Min Diversification (FMMD).}
Let the set $V$ be divided into $C$ disjoint groups $V_1, \ldots, V_C$ by a sensitive categorical attribute, such as sex or race.
For FMMD, the fairness-aware variant of MMD, the group fairness constraints restrict the selection of items from each group $V_c$ for $c \in [C]$ so that the number of items selected from $V_c$ lies within a range of values from $l_c$ to $h_c$ (both inclusive).
Meanwhile, it also requires that the total number of selected items is $k$.
Formally, the collection $\mathcal{F}$ of all feasible solutions for FMMD is
\begin{equation*}
  \mathcal{F} = \{ S \subseteq V : |S| = k \wedge l_c \leq  |S \cap V_c| \leq h_c, \forall c \in [C] \}
\end{equation*}
and, to discard from consideration trivially empty sets $\mathcal{F}$, we will further assume that $l_c \leq h_c \leq |V_c|$ and $\sum_{c=1}^{C} l_c \leq k \leq \sum_{c=1}^{C} h_c$.
The FMMD problem asks for a subset $S$ of $V$ so that $S$ satisfies the group fairness constraints (i.e., $S \in \mathcal{F}$) and $div(S)$ is maximized, or formally, $S^* = \argmax_{S \in \mathcal{F}}{div(S)}$,
where $S^*$ and $\mathtt{OPT} = div(S^*)$ denote the optimal solution of FMMD and its diversity value, respectively.

\paragraph{Hardness of FMMD.}
By using a reduction from the \textsc{Clique} problem, MMD is proven to be NP-hard for general metric spaces and cannot be approximated within any factor greater than $\frac{1}{2}$ unless P=NP~\cite{ERKUT199048, DBLP:journals/ior/RaviRT94}.
Nevertheless, a greedy algorithm provides the best possible $\frac{1}{2}$-approximate solution in $O(k n)$ time~\cite{DBLP:journals/tcs/Gonzalez85}.
Although the greedy algorithm does not work for FMMD directly, as it may provide solutions that do not fall within $\mathcal{F}$ (i.e., are not ``fair''), it will be used as a subroutine for our \textsf{FMMD-S} algorithm in Section~\ref{subsec:alg2} for data reduction.
Since MMD is a special case of FMMD when $C=1$ and $l_1 \leq k \leq h_1$, the hardness result for MMD can be generalized to FMMD as follows:
\begin{theorem}
  FMMD is NP-hard and cannot be approximated by a factor of $\frac{1}{2} + \varepsilon$ for any parameter $\varepsilon > 0$ unless P=NP.
\end{theorem}

\section{Algorithms}
\label{sec:alg}

In this section, we describe our proposed algorithms for FMMD. First, we propose \textsf{FMMD-E}, an exact algorithm that runs in $O(n^k \log n)$ time (Section~\ref{subsec:alg1}). Second, we propose \textsf{FMMD-S}, a $\frac{1-\varepsilon}{5}$-approximation algorithm that runs in $O \big( C k n + C^{k} \log{\frac{1}{\varepsilon}} \big)$ time for any error parameter $\varepsilon \in (0,1)$ (Section~\ref{subsec:alg2}).

\subsection{An Exact ILP-Based Algorithm}
\label{subsec:alg1}

\begin{figure}
  \centering
  \subcaptionbox{FMMD\label{fig:2a}}[.4\linewidth]{\includegraphics[height=1.5in]{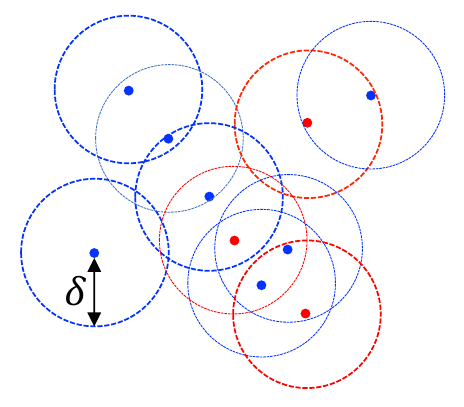}}
  \subcaptionbox{FIS\label{fig:2b}}[.4\linewidth]{\includegraphics[height=1.5in]{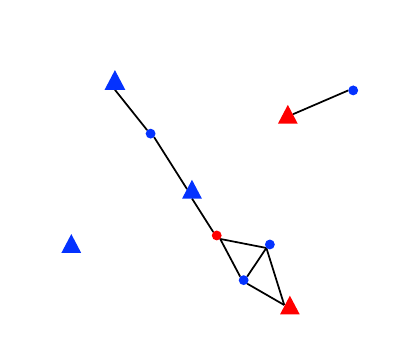}}
  \caption{Example for the reduction from Fair Max-Min Diversification (FMMD) to Fair Independent Set (FIS) on a dataset with $n = 10$ points and $C = 2$ groups in blue and red. An FMMD instance with $k = 5$, $l_c = 2$ and $h_c = 3$ for $c = 1,2$ is reduced to an FIS instance where a fair independent set of vertices is represented by triangles.}\label{fig:2}
\end{figure}

To build an exact algorithm for FMMD, we use ideas similar to~\cite{DBLP:conf/faw/AkagiAHNOOSUUW18} for the reduction from MMD to maximum independent set (MIS).
Given the set of feasible solutions $\mathcal{F}$ and a positive real number $\delta$, the decision version of FMMD asks whether there is a set $S\subseteq V$ such that $S \in \mathcal{F}$ and $div(S) \geq \delta$.
Given an instance of the FMMD decision problem, we build an undirected graph $G=(V, E)$ as follows: the set of vertices in $G$ is identical to $V$ and there is an edge between two vertices $u, v \in V$ if and only if $d(u, v) < \delta$.
We remind that a vertex set $S$ is called an \emph{independent set} iff no two vertices in $S$ are adjacent.
Moreover, we define the \emph{Fair Independent Set} (FIS) problem that determines whether there exists an independent vertex set $S \in \mathcal{F}$ on the graph $G$.
Based on the above definitions, the lemma below asserts the equivalence between FMMD and FIS.
\begin{lemma}\label{lm:fmmd:ivs}
  FMMD is equivalent to FIS.
\end{lemma}
\begin{proof}
  In the one direction, assume that the answer to FMMD is `\emph{yes}' -- i.e., there is a subset $S \in \mathcal{F}$ of $V$ with $div(S) \geq \delta$.
  Then, we have $d(u, v) \geq \delta$ for any $u, v \in S$.
  Thus, by construction, there is no edge $(u, v) \in E$, and $S$ is an independent vertex set of $G$.
  Therefore, the answer to FIS is `\emph{yes}' as well. 
  In the opposite direction, assume that the answer to FIS is `\emph{yes}' -- i.e., $ S \in \mathcal{F} $ is an independent set.
  By definition, there is no edge between any of its vertices in $G$, which by construction means that $d(u, v) \geq \delta$ for any $u, v \in S$ and, therefore, we have $div(S) \geq \delta$ for the given $S \in \mathcal{F}$.
  Therefore, the answer to FMMD is also `\emph{yes}'. 
  We thus prove that the answer to FMMD is `\emph{yes}' if and only if the answer to FIS is `\emph{yes}', which concludes the proof.
\end{proof}

Additionally, we have two observations for FMMD, which are easy to verify from its definition.
\begin{observation}[Monotonicity]\label{lm:fmmd:mono}
  If there exists a set $S \in \mathcal{F}$ with $div(S) \geq \delta$, then such a set will exist for any $\delta' \leq \delta$; If there does not exist any set $S \in \mathcal{F}$ with $div(S) \geq \delta$, then such a set will not exist for any $\delta' \geq \delta$.
\end{observation}
\begin{observation}[Discontinuity]\label{lm:fmmd:dis}
  The optimal diversity value $\mathtt{OPT}$ for FMMD is always equal to the distance $d(u, v)$ between some pair of items $u, v \in V$.
\end{observation}

From all the above results, the following theorem asserts that FMMD is reducible to FIS.
\begin{theorem}
\label{theorem:reduction}
  An exact solution of FMMD is obtained by solving $O(\log{n})$ FIS instances.
\end{theorem}
\begin{proof}
  Let us consider the following algorithm.
  First, compute and sort the distances between all pairs of items in $V$.
  Then, use a binary search on the sorted array of pairwise distances to find the largest $d^*$ such that the answer to its corresponding FIS instance is `\emph{yes}'.
  The binary search finds $d^*$ in $O(\log{n})$ steps, as the number of pairwise distances is $O(n^2)$.
  And it holds that $div(S^*) \geq d^*$ from Lemma~\ref{lm:fmmd:ivs}. Observation~\ref{lm:fmmd:mono} guarantees that there does not exist any $S \in \mathcal{F}$ with $div(S) > d^*$ due to the maximality of $d^*$. Observation~\ref{lm:fmmd:dis} ensures that $d^*$ is exactly equal to $\mathtt{OPT}$. Thus, the above procedure identifies the exact solution to FMMD.
\end{proof}

\paragraph{The ILP Formulation of FMMD.}
In light of Theorem~\ref{theorem:reduction}, what remains to obtain an exact algorithm for FMMD is to design an exact algorithm for FIS. We note that FIS without fairness constraints is equivalent to the \emph{maximum independent set} (MIS) problem. We thus adapt the edge-based integer-linear programming (ILP) formulation of MIS by adding fairness constraints to define an FIS instance, as shown in Eq.~\ref{ilp:obj}--\ref{ilp:constr:4}.
\begin{align}
  \max          \quad & z = \sum_{i=1}^{n} x_i \label{ilp:obj}\\
  \mathrm{s.t.} \quad & x_i + x_j \leq 1, \forall  (v_i, v_j) \in E \label{ilp:constr:1}\\
                      & \sum_{i=1}^{n} x_i \leq k \label{ilp:constr:2}\\
                      & l_c \leq \sum_{v_i \in V_c} x_i \leq h_c, \forall c \in [C] \label{ilp:constr:3}\\
                      & x_i \in \{0, 1\}, \forall i \in [n] \label{ilp:constr:4}
\end{align}
where $x_i$ is a binary variable to indicate whether $v_i \in V$ is included in the solution or not, the objective function in Eq.~\ref{ilp:obj} and the first constraint in Eq.~\ref{ilp:constr:1} are the same as the edge-based ILP formulation of MIS, the second constraint in Eq.~\ref{ilp:constr:2} limits the solution size to at most $k$, and the third constraint in Eq.~\ref{ilp:constr:3} is on the upper and lower bounds of the number of items chosen from each group $V_c$. By solving the ILP in Eq.~\ref{ilp:obj}--\ref{ilp:constr:4} optimally, we will either find a fair independent set $S = \{ v_i \in V : x_i = 1, i \in [n] \}$ of $G$ if $z = k$ or confirm that there does not exist such a set if $z < k$.

\begin{algorithm}[tb]
  \caption{\textsf{FMMD-E}}
  \label{alg:fmmd:e}
  \small
  \begin{algorithmic}[1]
    \REQUIRE Dataset $V = \bigcup_{c=1}^{C} V_c$ with $n = |V|$; lower and upper bounds $l_c, h_c \in \mathbb{Z}^{+}$ for $c \in [C]$; size constraint $k \in \mathbb{Z}^{+}$.
    \ENSURE A set $S^* \subseteq V$ such that $S^* \in \mathcal{F}$.
    \STATE Compute the distances of all pairs of items in $V$ and sort them ascendingly as $D[1, \ldots, \frac{n(n-1)}{2}]$
    \STATE Let $L \gets 1$, $H \gets \frac{n(n-1)}{2}$, $cur \gets \frac{L + H}{2}$, $S^* \gets \emptyset$
    \WHILE{$H > L$}
      \STATE Build an undirected graph $G=(V,E)$ where $E = \{(u, v) \in V \times V \;|\; d(u, v) < D[cur] \}$
      \STATE Compute the solution $\mathbf{x}$ of the ILP in Eq.~\ref{ilp:obj}--\ref{ilp:constr:4}
      \STATE Find a set $S = \{ v_i \in V : x_i = 1, i \in [n] \}$ based on $\mathbf{x}$
      \IF {$|S| = k$}
        \STATE If $div(S) > div(S^*)$ or $S^* = \emptyset$, then $S^* \gets S$
        \STATE Let $L \gets cur + 1$ and $cur \gets \frac{L + H}{2}$
      \ELSE
        \STATE Let $H \gets cur - 1$ and $cur \gets \frac{L + H}{2}$
      \ENDIF
    \ENDWHILE
    \STATE \textbf{return} $S^*$
  \end{algorithmic}
\end{algorithm}

\paragraph{Algorithm Description and Complexity.}
By combining the constructive proof of Theorem~\ref{theorem:reduction} and the ILP formulation of FMMD, we obtain \textsf{FMMD-E}, an exact algorithm for FMMD, as presented in Algorithm~\ref{alg:fmmd:e}.
First, it computes the distances of all $\frac{n(n-1)}{2}$ pairs of distinct items in $V$ in $O(n^2)$ steps and sorts them in ascending order in an array $D[1, \ldots, \frac{n(n-1)}{2}]$ in $O(n^2 \log n)$ steps.
Then, a binary search is performed on $D$ to find $\mathtt{OPT}$ in $O(\log n)$ steps.
For each guess $D[cur]$ of $\mathtt{OPT}$, it builds an undirected graph $G$ in $O(n^2)$ steps and finds a set $S$ by solving the ILP in Eq.~\ref{ilp:obj}--\ref{ilp:constr:4} in $\binom{n}{k} = O(n^k)$ steps.
If $|S| = k$, then $S \in \mathcal{F}$ and $div(S) \geq D[cur]$. In this case, the search space is narrowed to the upper half to check whether there is a better solution. Otherwise, or if $|S| < k$, then $\mathtt{OPT} < D[cur]$ and the search space is narrowed to the lower half.
Finally, when the binary search is terminated, the algorithm has found the exact solution $S^*$ to FMMD.
The time complexity of \textsf{FMMD-E} is $O(n^k \log n)$. Moreover, since $|D| = |E| = O(n^2)$, its space complexity is $O(n^2)$.

\subsection{A More Scalable Approximation Algorithm}
\label{subsec:alg2}

\begin{algorithm}[tb]
  \caption{\textsf{FMMD-S}}
  \label{alg:fmmd:s}
  \small
  \begin{algorithmic}[1]
    \REQUIRE Dataset $V = \bigcup_{c=1}^{C} V_c$ with $n = |V|$; lower and upper bounds $l_c, h_c \in \mathbb{Z}^{+}$ for $c \in [C]$; size constraint $k \in \mathbb{Z}^{+}$; error parameter  $\varepsilon \in (0,1)$.
    \ENSURE A set $S \subseteq V$ such that $S \in \mathcal{F}$.
    \STATE Pick an arbitrary item $u_1$ from $V$ and set $U = \{u_1\}$\label{ln:gmm:s}\label{line:startinit}
    \FOR{$i \gets 2,\ldots,k$}
      \STATE $u_i \gets \argmax_{v \in V} \min_{u \in U} d(u, v)$ and $U \gets U \cup \{u_i\}$\label{ln:gmm:t}
    \ENDFOR \label{line:endinit}
    \STATE Set $U_c \gets V_c \cap U$, $d' \gets 2 \cdot div(U)$\label{line:partition}
    \REPEAT \label{line:startfill}
      \FOR{$c \gets 1,\ldots,C$}
        \WHILE{$|U_c| < k$ and there exists some item $v \in V_c$ such that $div(U_c \cup \{v\}) \geq d'$\label{ln:fmmd:s:col}}
          \STATE $u'_{c} \gets \argmax_{v \in V_c} \min_{u \in U_c} d(u, v)$
          \STATE $U_c \gets U_c \cup \{u'_{c}\}$
        \ENDWHILE
      \ENDFOR \label{line:endfill}
      \STATE Build an undirected graph $G=(V',E)$, where $V' = \bigcup_c U_c$ and $E = \{(u, v) \in V' \times V' \;|\; d(u, v) < \frac{d'}{2}\}$
      \STATE Compute the solution $\mathbf{x}$ of the ILP in Eq.~\ref{ilp:obj}--\ref{ilp:constr:4} \label{line:startexact}
      \STATE Find a set $S = \{ v_i \in V : x_i = 1, i \in [n] \}$ based on $\mathbf{x}$\label{line:endexact}
      \IF {$|S| < k$}\label{line:startsolution}
        \STATE $S \gets \emptyset$ and $d' \gets (1-\varepsilon) \cdot d'$
      \ENDIF
    \UNTIL{$S \neq \emptyset$} 
    \STATE \textbf{return} $S$\label{line:endsolution}
  \end{algorithmic}
\end{algorithm}

The main drawback of \textsf{FMMD-E} is that it cannot handle large datasets due to exponential complexity. Standard optimization libraries, such as CPLEX\footnote{\url{www.ibm.com/products/ilog-cplex-optimization-studio}} and Gurobi\footnote{\url{www.gurobi.com/products/gurobi-optimizer/}}, can only solve ILPs with up to several thousand variables optimally in a reasonable time. A natural approach to addressing this challenge is to identify a ``\emph{coreset}'', i.e., a small subset of the original dataset on which the exact algorithm is run to look for approximate solutions.
Formally, a subset $V' \subseteq V$ is called an $\alpha$-coreset ($0 \leq \alpha \leq 1$) of $V$ for FMMD if $\mathtt{OPT}[V'] \geq \alpha \cdot \mathtt{OPT}$, where $\mathtt{OPT}[V']$ is the optimal diversity value for FMMD on $V'$.

It now remains to answer \emph{i)} how such a coreset is built and \emph{ii)} what approximation factor is obtained.
For \emph{i)}, we are inspired by the notion of \emph{composable coresets}~\cite{DBLP:conf/aaai/ZadehGMZ17,DBLP:conf/pods/IndykMMM14} for MMD in streaming and distributed settings. The basic idea is first to partition the dataset and run the greedy algorithm of~\cite{DBLP:journals/tcs/Gonzalez85} on each partition to obtain a partial solution and then compute a final solution from the union of partial solutions. In the context of FMMD, the dataset is naturally divided into $C$ groups. Thus, we first find a solution from each group, then consider the union of all group-specific solutions as our \emph{coreset}, and finally use \textsf{FMMD-E} to obtain a solution from the coreset, which is feasible since the coreset size is small. 
We refer to the resulting algorithm as \textsf{FMMD-S}.
For \emph{ii)}, we prove that the obtained solution offers an approximation factor of $\frac{1-\varepsilon}{5}$ for any $\varepsilon \in (0,1)$.

\paragraph{Algorithm Description.} 
\textsf{FMMD-S} is described in Algorithm~\ref{alg:fmmd:s}. 
Initially, it invokes the greedy algorithm on $V$ without fairness constraints to compute an initial solution $U$ (Lines~\ref{line:startinit}-\ref{line:endinit}). 
Note that the greedy algorithm is $\frac{1}{2}$-approximate for MMD, and any feasible solution of FMMD must also be feasible for MMD. 
Therefore, the optimal diversity $\mathtt{OPT}$ of FMMD is bounded by $2 \cdot div(U)$.
Subsequently, the algorithm divides $U$ by group into $U_1, \ldots, U_C$ and guesses $\mathtt{OPT}$ as its upper bound $d' = 2 \cdot div(U)$ (Line~\ref{line:partition}). 
For each $c \in [C]$, it runs the greedy algorithm to add new items from $V_c$ to $U_c$ until $|U_c| = k$ or there does not exist any $v \in V_c$ to make $div(U_c \cup \{v\}) \geq d'$ (Lines~\ref{line:startfill}-\ref{line:endfill}). 
At this point, each $U_c$ is a partial group-specific solution, and the union $V' = \bigcup_c U_c$ of partial solutions is the \emph{coreset}.
Next, using a similar procedure to \textsf{FMMD-E}, it builds a graph $G$ on $V'$ with diversity threshold $\frac{d'}{2}$ and solves the ILP of Eq.~\ref{ilp:obj}--\ref{ilp:constr:4} on $G$ to obtain a solution $S$ (Lines~\ref{line:startexact}-\ref{line:endexact}). Finally, if $|S| = k$, we have got a solution $S \in \mathcal{F}$ with $div(S) \geq \frac{d'}{2}$ and $S$ will be returned as the final solution; otherwise, $d'$ is decreased by a factor of $1-\varepsilon$, where $\varepsilon \in (0,1)$ is an error parameter, and the above procedure is executed again for the smaller $d'$ until a feasible solution $S$ is found (Lines~\ref{line:startsolution}-\ref{line:endsolution}).

\paragraph{Theoretical Analysis.}
Next, we give the complexity and approximation factor of \textsf{FMMD-S}.

\begin{theorem}\label{thm:fmmd:s:approx}
  \textnormal{\textsf{FMMD-S}} is a $\frac{1-\varepsilon}{5}$-approximation algorithm for FMMD running in $O \big( C k n + C^{k} \log{\frac{1}{\varepsilon}} \big)$ time.
\end{theorem}
\begin{proof}
  If there is any set $S' \subseteq V'$ s.t.~$S' \in \mathcal{F}$ and $div(S') \geq \frac{d'}{2}$, then \textsf{FMMD-S} identifies such $S'$ (Line~\ref{line:endexact}) from the exact solution of the ILP in Eq.~\ref{ilp:obj}--\ref{ilp:constr:4}. 
  In addition, since the greedy algorithm (Lines~\ref{line:startinit}-\ref{line:endinit}) is $\frac{1}{2}$-approximate~\cite{DBLP:journals/ior/RaviRT94}, the initial value of $d'$ is at least $2\cdot\frac{\mathtt{OPT}}{2} = \mathtt{OPT} > \frac{2}{5} \cdot \mathtt{OPT}$.
  Therefore, to prove the approximation factor, it suffices to show that there exists some $S' \subseteq V'$ s.t.~$S' \in \mathcal{F}$ and $div(S') \geq \frac{d'}{2}$ when $d' \in [\frac{2(1-\varepsilon)}{5} \cdot \mathtt{OPT}, \frac{2}{5} \cdot \mathtt{OPT}]$. 
  
  Towards this end, we next construct such a set $S'$ from $V'$. Let $S^*$ be the optimal solution for FMMD on $V$, and $S^*_c = S^* \cap V_c$ be its subset from group $c$. First, we initialize $S' = \emptyset$. Then, we consider two cases for different groups. We consider first the groups of Case \#1 in arbitrary order, then those of Case \#2 in arbitrary order, and select $|S^*_c|$ items from each group $c$ into $S'$.
  
  \underline{Case \#1} ($|U_c| < k$): Let $f: V_c \rightarrow U_c$ map each item $v \in V_c$ to its nearest neighbor $f(v)$ in $U_c$. 
  Note that the condition in Line~\ref{ln:fmmd:s:col} ensures that $d(v, u) < d'$ for any $v \in V_c$ and $u\in U_c$. 
  For each item $s_{c, i} \in S^*_c$, we add item $f(s_{c, i})$ into $S'$. 
  We now show that the added items are distinct.
  Indeed, if $f(s_{c, i}) \equiv f(s_{c, j})$ for $i \neq j$, then the triangle inequality would give $d(s_{c, i}, s_{c, j}) \leq d(s_{c, i},f(s_{c, i})) + d(s_{c, j},f(s_{c, j})) = d(s_{c, i},f(s_{c, i})) + d(s_{c, j},f(s_{c, i})) < 2\cdot d' \leq \frac{4}{5} \cdot \mathtt{OPT} < \mathtt{OPT}$; however, at the same time we have $d(s_{c, i}, s_{c, j}) \geq \mathtt{OPT}$ because $s_{c, i}, s_{c, j} \in S^*$, which leads to a contradiction. 
  Moreover, because we have identified for each $s_{c, i} \in S^*_c$ one distinct item in $U_c$, we have $|U_c| \geq |S^*_c|$. 
  After processing all the groups in {Case \#1}, we have $d(f(s^*_i), f(s^*_j)) \geq d(s^*_i, s^*_j) - d(s^*_i, f(s^*_i)) - d(s^*_j, f(s^*_j)) > \mathtt{OPT} - 2 d' \geq \frac{\mathtt{OPT}}{5}$ for any $f(s^*_i), f(s^*_j) \in S'$ and thus $div(S') > \frac{\mathtt{OPT}}{5}$.
  
  \underline{Case \#2} ($|U_c| = k$): Let $g: U_c \rightarrow S'$ map each item $u \in U_c$ to its nearest neighbor $g(u)$ in the current instance of $S'$. 
  We remove from $U_c$ every $u \in U_c$ with $d(u, g(u)) < \frac{d'}{2}$.
  Because the condition of Line~\ref{ln:fmmd:s:col} ensures $d(u_{c, i}, u_{c, j}) \geq d'$ for any $i \neq j$, there is at most one item removed for each item in $S'$ -- otherwise, the triangle inequality would give $d(u_{c, i}, u_{c, j}) < d'$, thus leading to a contradiction.
  Therefore, at least $k - |S'|$ items remain in $U_c$.
  Moreover, $|S^*_c| \leq k - |S'|$ because $S'$ always contains the same number of items from each considered group as $S^*$ throughout the construction process.
  We pick $|S^*_c|$ items from the remaining ones and add them to $S'$.
  After this operation, we still have $div(S') \geq \frac{d'}{2} \geq \frac{1-\varepsilon}{5} \cdot \mathtt{OPT}$ since our earlier removal of items from $U_c$ ensured $d(u_{c, i}, g(u_{c, i})) \geq \frac{d'}{2}$ for each added $u_{c, i}$.
  Finally, after processing all groups in {Case \#2}, we get a set $S'$ that contains the same number of items from each group $c \in [C]$ as $S^*$, which implies that $S' \in \mathcal{F}$, and $div(S') \geq \frac{1-\varepsilon}{5} \cdot \mathtt{OPT}$.
  Therefore, we conclude that \textsf{FMMD-S} is a $\frac{1-\varepsilon}{5}$-approximation algorithm for FMMD.
  
  Since it takes $O(n k)$ time to compute $U$ as well as $U_c$ for each $c \in [C]$, the total time to compute $V'$ is $O(C k n)$ and $|V'| \leq Ck$. Then, the time to solve the ILP in Eq.~\ref{ilp:obj}--\ref{ilp:constr:4} for \textsf{FMMD-S} is $O(C^k)$ because there are at most $\binom{Ck}{k} = O(C^k)$ possible solutions to consider. Moreover, the number of iterations for $d'$ is $O( \log{\frac{1}{\varepsilon}})$ since the ratio between the first and last values of $d'$ is $O(1)$. Thus, the time complexity of \textsf{FMMD-S} is $O \big( C k n + C^{k} \log{\frac{1}{\varepsilon}} \big)$. When $k = o(\log{n})$ and $C = O(1)$, its time complexity is reduced to $O\big( n(k + \log{\frac{1}{\varepsilon}}) \big)$. Additionally, its space complexity is $O(n + C^2 k^2)$ since the number of edges in $G$ is $O(C^2 k^2)$.
\end{proof}

\section{Experimental Evaluation}
\label{sec:exp}

\subsection{Experimental Setup}
In this section, we conduct extensive experiments to evaluate the performance of our proposed algorithms, i.e., \textsf{FMMD-E} and \textsf{FMMD-S}.
We compare them with the state-of-the-art FMMD algorithms, including \textsf{FairSwap}, \textsf{FairFlow}, and \textsf{FairGMM} in~\cite{DBLP:conf/icdt/Moumoulidou0M21}, \textsf{FairGreedyFlow} in~\cite{DBLP:conf/icdt/Addanki0MM22}, and \textsf{SFDM1} and \textsf{SFDM2} in~\cite{DBLP:conf/icde/WangFM22}.
We implemented all the above algorithms in Python 3 using the NetworkX library for building and manipulating graphs and the Gurobi optimizer for solving ILPs.
All the experiments were carried out on a desktop with an Intel Core i5-9500 3.0GHz processor and 32GB RAM running Ubuntu 20.04.3 LTS. Each algorithm was run on a single thread.
All data and code are publicly available at \url{https://osf.io/te34m/}.

We use four public real-world datasets listed in Table~\ref{tab:data}, where $dim$ is the dimensionality of the feature vector.
The detailed information and preprocessing procedures on each dataset are described in Appendix~\ref{sec:datasets}.
The fairness constraints in our experiments are defined according to the \emph{proportional representation}~\cite{DBLP:conf/ijcai/CelisHV18,DBLP:conf/nips/HalabiMNTT20}: For each group $c \in [C]$, we set $l_c = \max(1, (1-\alpha) k \cdot \frac{|V_c|}{n})$ and $h_c = (1+\alpha) k \cdot \frac{|V_c|}{n}$ with $\alpha = 0.2$ in \textsf{FMMD-E} and \textsf{FMMD-S} and $k_c = \lceil k \cdot \frac{|V_c|}{n} \rceil$ or $\lfloor k \cdot \frac{|V_c|}{n} \rfloor$ in all other algorithms. All the algorithms were executed ten times in each experiment. We report the average running time and average diversity value of the solutions provided by each algorithm. We use `N/A' to indicate that an algorithm either does not find a solution within one day or does not work when $C>2$ (i.e.,~\textsf{FairSwap} and \textsf{SFDM1}).
In the preliminary experiments (see Appendix~\ref{sec:add:exp}), we find that the solution quality of FMMD-S hardly improves when $\varepsilon$ is decreased below $0.05$ and so we fix $\varepsilon = 0.05$ for FMMD-S in all the remaining experiments.

\begin{table}[tb]
\centering
\scriptsize
\caption{Statistics of datasets used in our experiments}
\label{tab:data}
\begin{tabular}{cccccc}
\toprule
\textbf{Dataset} & \textbf{Group} & $C$  & $n$ & $dim$ & \textbf{Distance Metric} \\
\midrule
\multirow{3}{*}{Adult} & Sex & 2 & \multirow{3}{*}{48,842} & \multirow{3}{*}{6} & \multirow{3}{*}{$l_2$-distance} \\
 & Race & 5 & & & \\
 & S+R & 10 & & & \\
\midrule
\multirow{3}{*}{CelebA} & Sex & 2 & \multirow{3}{*}{202,599} & \multirow{3}{*}{25,088} & \multirow{3}{*}{$l_1$-distance} \\
 & Age & 2 & & & \\
 & S+A & 4 & & & \\
\midrule
\multirow{3}{*}{Census} & Sex & 2 & \multirow{3}{*}{2,426,116} & \multirow{3}{*}{25} & \multirow{3}{*}{$l_1$-distance} \\
 & Age & 7 & & & \\
 & S+A & 14 & & & \\
\midrule
Twitter & Sex & 3 & 18,836 & 1,024 & Angular distance \\
\bottomrule
\end{tabular}
\end{table}

\begin{table}[tb]
\scriptsize
\centering
\caption{Diversity values of the solutions returned by different algorithms on small datasets for solution size $k=10$. The optimums $\mathtt{OPT}^*$ without fairness constraints are reported to show ``the price of fairness''.}
\label{tab:res:small}
\begin{tabular}{ccrrrrrrrrr}
\toprule
\textbf{Dataset} & \textbf{Group} & \textsf{FairSwap} & \textsf{FairFlow} & \textsf{FairGMM} & \textsf{FairGreedyFlow} & \textsf{SFDM1} & \textsf{SFDM2} & \textsf{FMMD-E} & \textsf{FMMD-S} & $\mathtt{OPT}^*$ \\
\midrule
\multirow{3}{*}{Adult}  & Sex  & 4.51     & 3.24    & 4.81     & 2.18    & 4.03     & 4.18    & \textbf{5.30}     & 4.64     & \multirow{3}{*}{5.30} \\
                        & Race & N/A      & 1.73    & N/A      & 1.33    & N/A      & 2.83    & \textbf{4.54}     & 4.01     & \\
                        & S+R  & N/A      & 0.84    & N/A      & 0.99    & N/A      & 2.04    & \textbf{3.12}     & 2.88     & \\
\midrule
\multirow{3}{*}{CelebA} & Sex  & 101457.3 & 55540.7 & 127354.6 & 46266.9 & 94873.3  & 93216.6 & \textbf{129818.2} & 106959.3 & \multirow{3}{*}{129871.5} \\
                        & Age  & 110098.1 & 54649.6 & 127871.2 & 46312.3 & 102762.9 & 91578.2 & \textbf{129871.5} & 116701.6 & \\
                        & S+A  & N/A      & 42412.7 & N/A      & 39967.2 & N/A      & 88026.7 & \textbf{127974.6} & 108055.3 & \\
\midrule
\multirow{2}{*}{Census} & Sex  & 28.4     & 15.8    & 29.8     & 14.7    & 27.2     & 28.0    & \textbf{34.0}     & 30.3     & \multirow{2}{*}{35.0} \\
                        & Age  & N/A      & 7.6     & N/A      & 9.3     & N/A      & 15.7    & \textbf{34.0}     & 30.3     & \\
\midrule
Twitter                 & Sex  & N/A      & 1.23    & 1.44     & 1.23    & N/A      & 1.39    & \textbf{1.51}     & 1.46     & 1.51 \\
\bottomrule
\end{tabular}
\end{table}

\begin{figure}[t]
  \centering
  \includegraphics[width=0.9\textwidth]{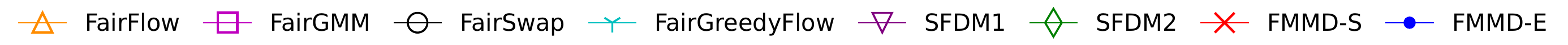}\\
  \subcaptionbox{Adult (Sex)}[.195\linewidth]{
    \includegraphics[height=1in]{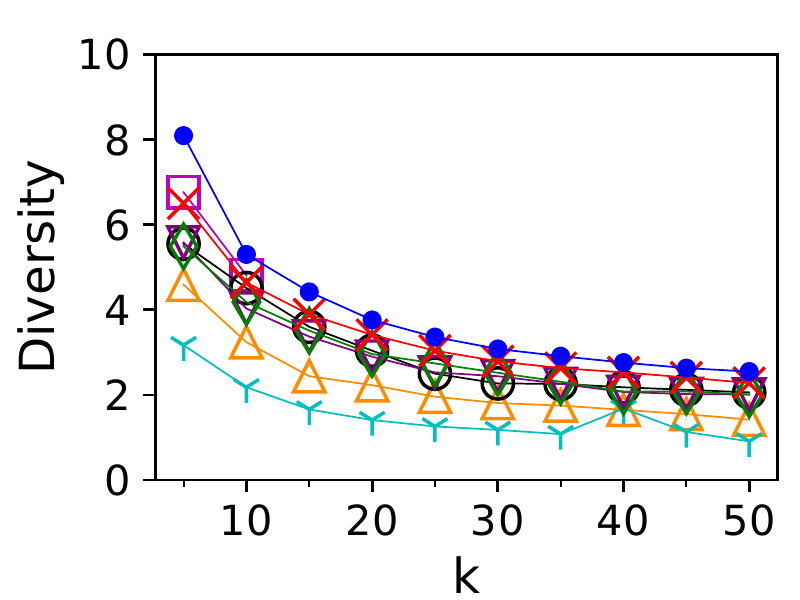}
  }
  \hfill
  \subcaptionbox{Adult (Race)}[.195\linewidth]{
    \includegraphics[height=1in]{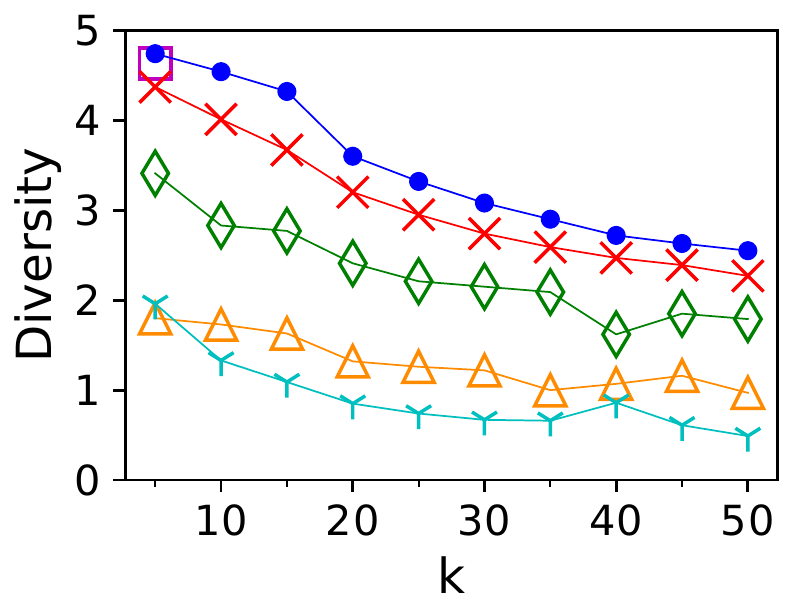}
  }
  \hfill
  \subcaptionbox{Adult (S+R)}[.195\linewidth]{
    \includegraphics[height=1in]{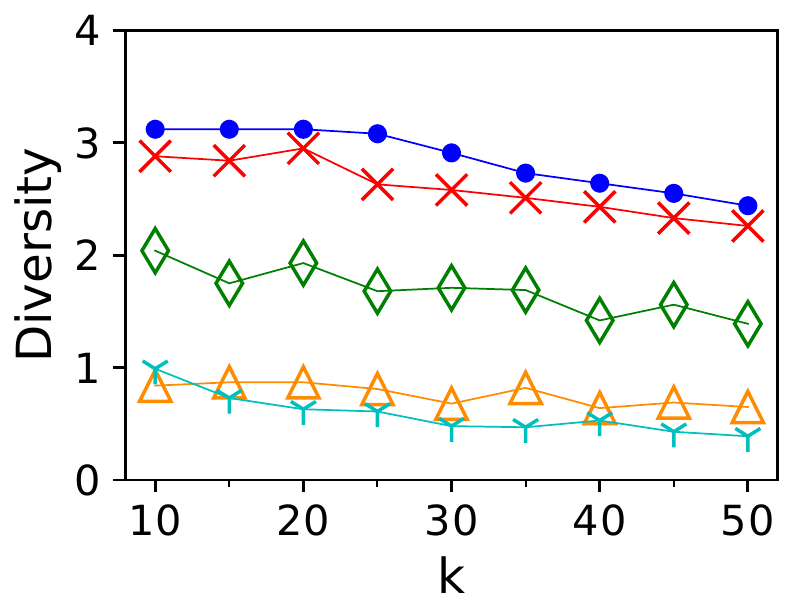}
  }
  \hfill
  \subcaptionbox{CelebA (Sex)}[.195\linewidth]{
    \includegraphics[height=1in]{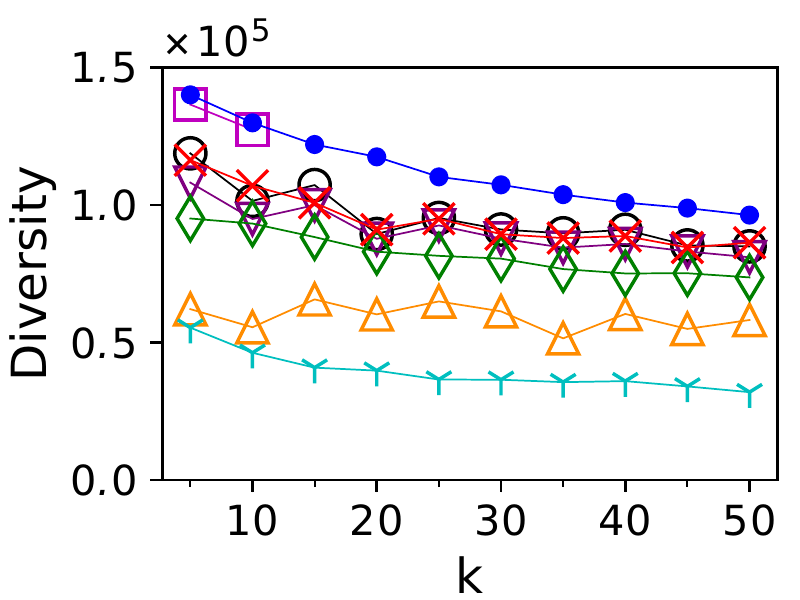}
  }
  \hfill
  \subcaptionbox{CelebA (Age)}[.195\linewidth]{
    \includegraphics[height=1in]{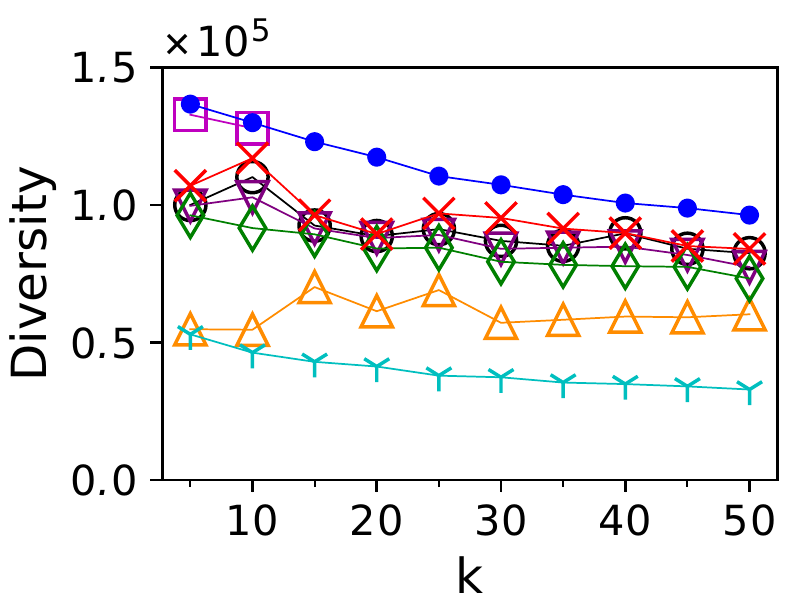}
  }
  \\
  \subcaptionbox{CelebA (S+A)}[.195\linewidth]{
    \includegraphics[height=1in]{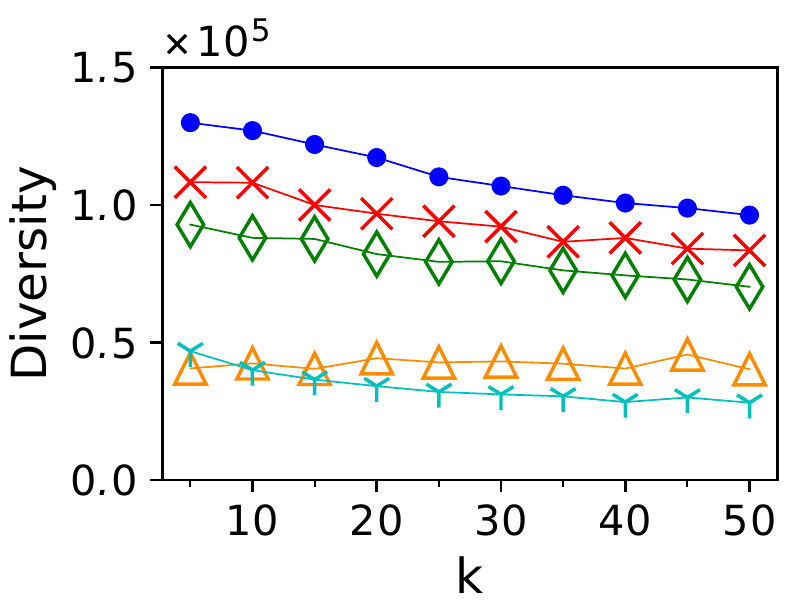}
  }
  \hfill
  \subcaptionbox{Census (Sex)}[.195\linewidth]{
    \includegraphics[height=1in]{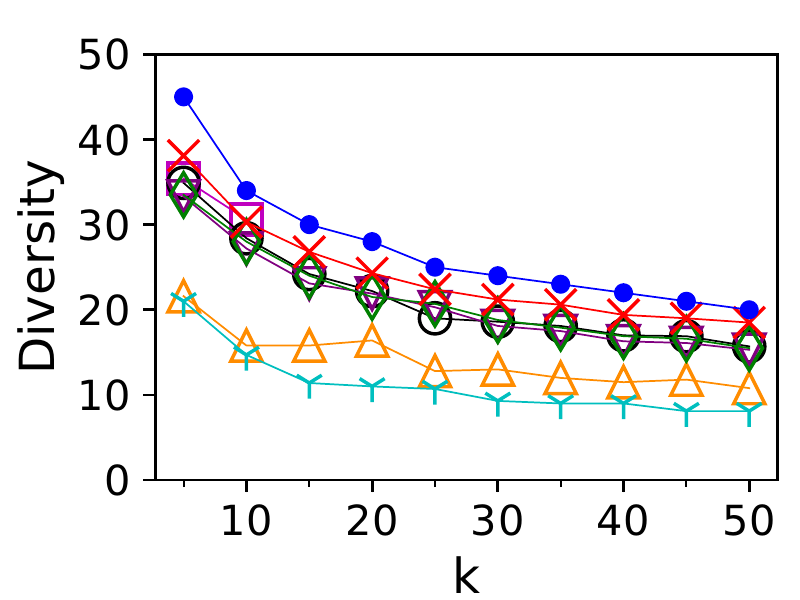}
  }
  \hfill
  \subcaptionbox{Census (Age)}[.195\linewidth]{
    \includegraphics[height=1in]{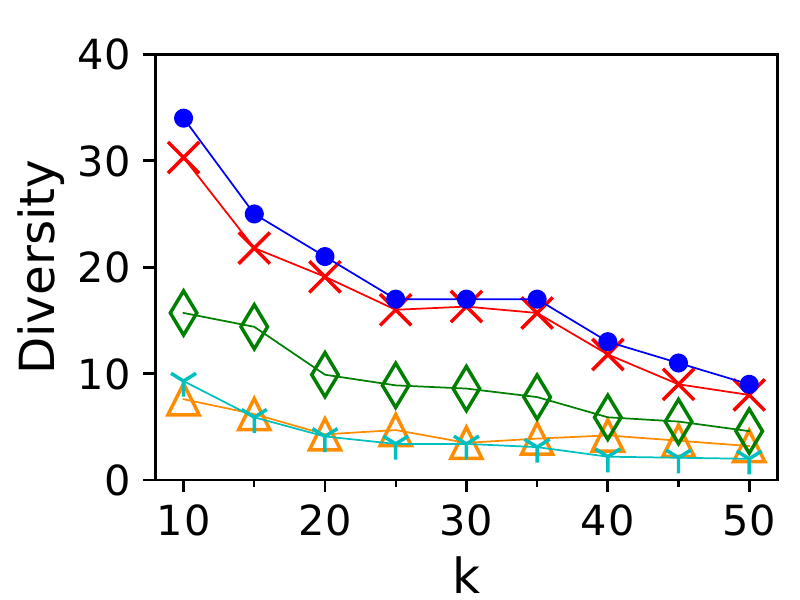}
  }
  \hfill
  \subcaptionbox{Census (S+A)}[.195\linewidth]{
    \includegraphics[height=1in]{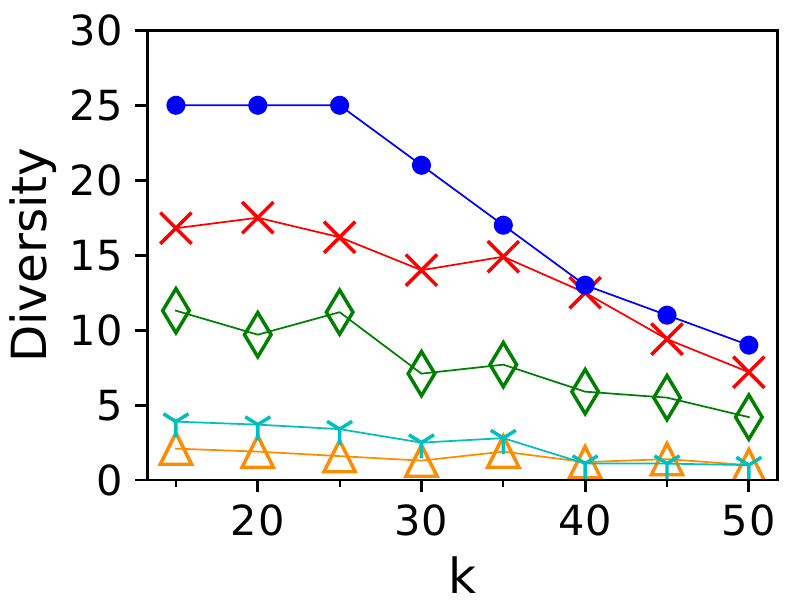}
  }
  \hfill
  \subcaptionbox{Twitter (Sex)}[.195\linewidth]{
    \includegraphics[height=1in]{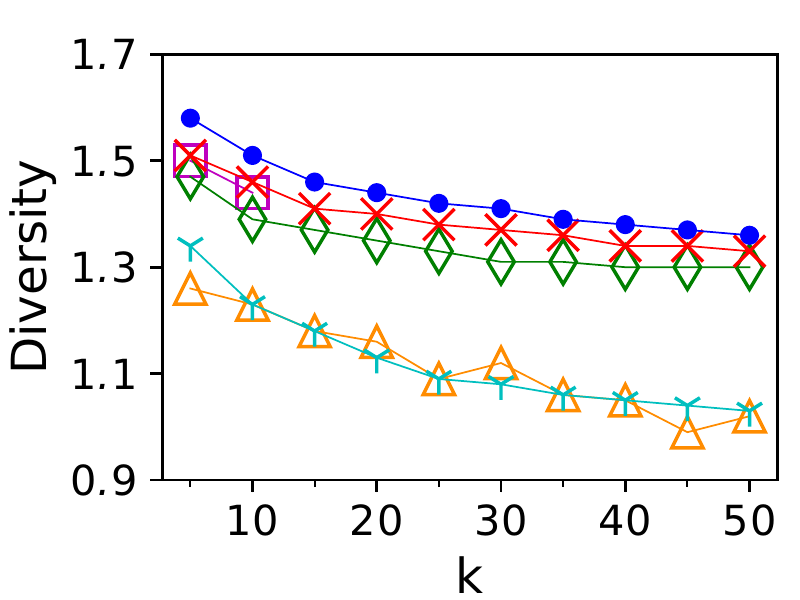}
  }
  \caption{Diversity values of the solutions of different algorithms with varying solution size $k$ on small datasets.}
  \label{fig:exp:div}
\end{figure}

\begin{figure}[t]
  \centering
  \includegraphics[width=0.9\textwidth]{figure/small-k/legend.pdf}\\
  \subcaptionbox{Adult (Sex)}[.195\linewidth]{
    \includegraphics[height=1in]{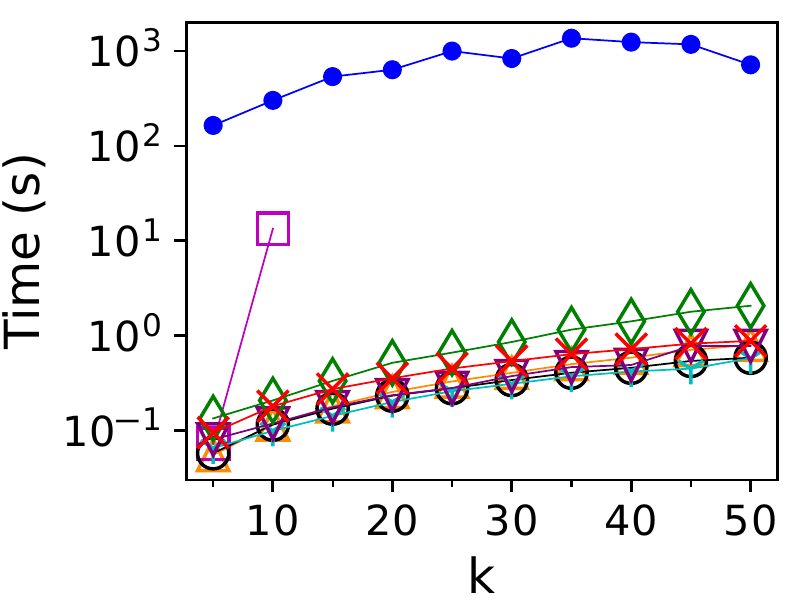}
  }
  \hfill
  \subcaptionbox{Adult (Race)}[.195\linewidth]{
    \includegraphics[height=1in]{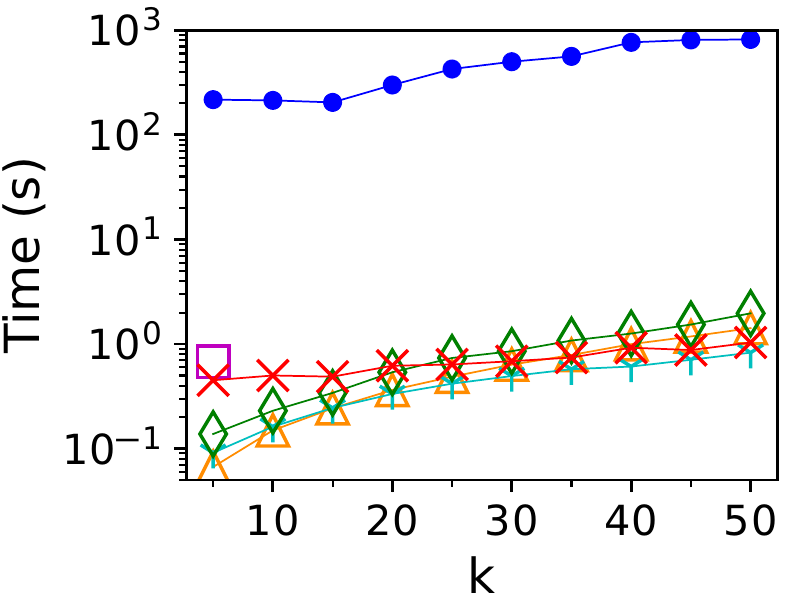}
  }
  \hfill
  \subcaptionbox{Adult (S+R)}[.195\linewidth]{
    \includegraphics[height=1in]{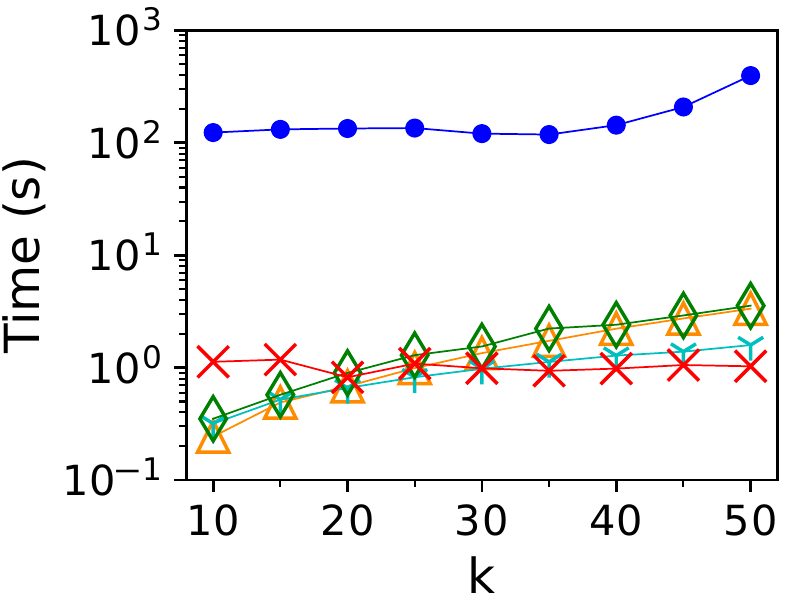}
  }
  \hfill
  \subcaptionbox{CelebA (Sex)}[.195\linewidth]{
    \includegraphics[height=1in]{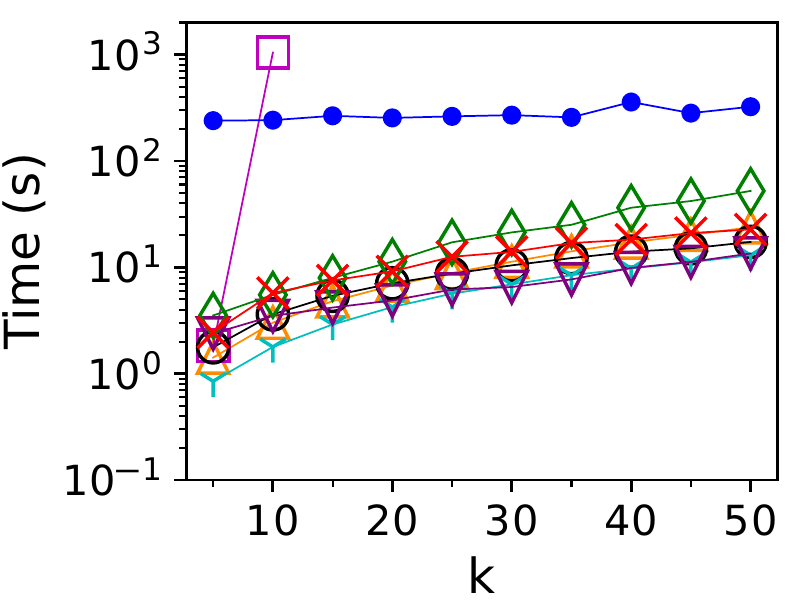}
  }
  \hfill
  \subcaptionbox{CelebA (Age)}[.195\linewidth]{
    \includegraphics[height=1in]{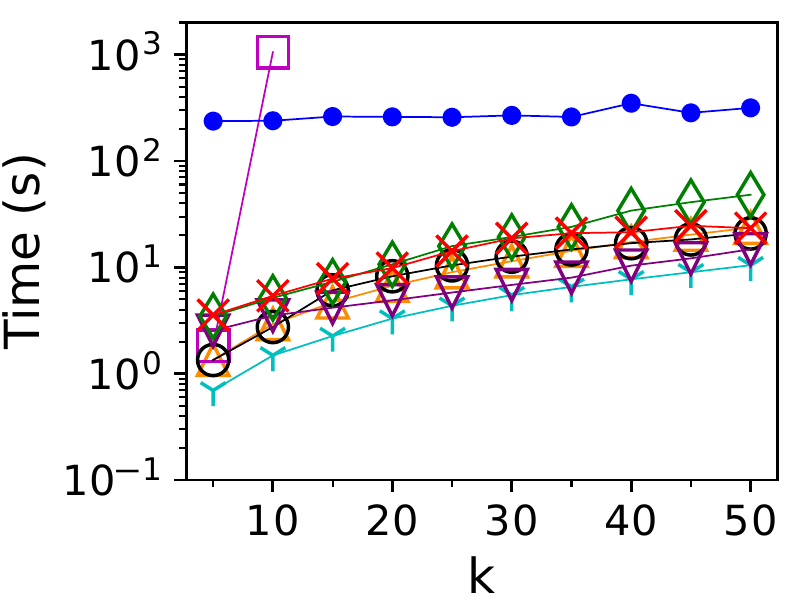}
  }
  \\
  \subcaptionbox{CelebA (S+A)}[.195\linewidth]{
    \includegraphics[height=1in]{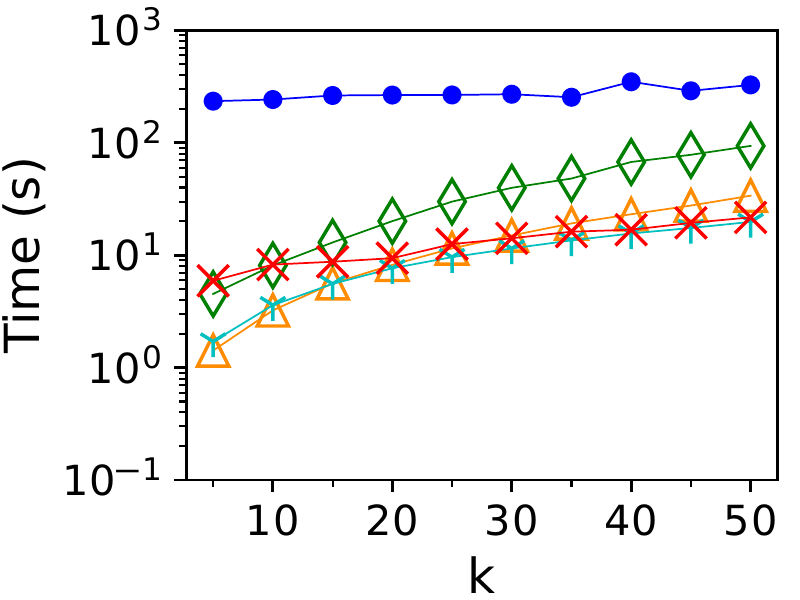}
  }
  \hfill
  \subcaptionbox{Census (Sex)}[.195\linewidth]{
    \includegraphics[height=1in]{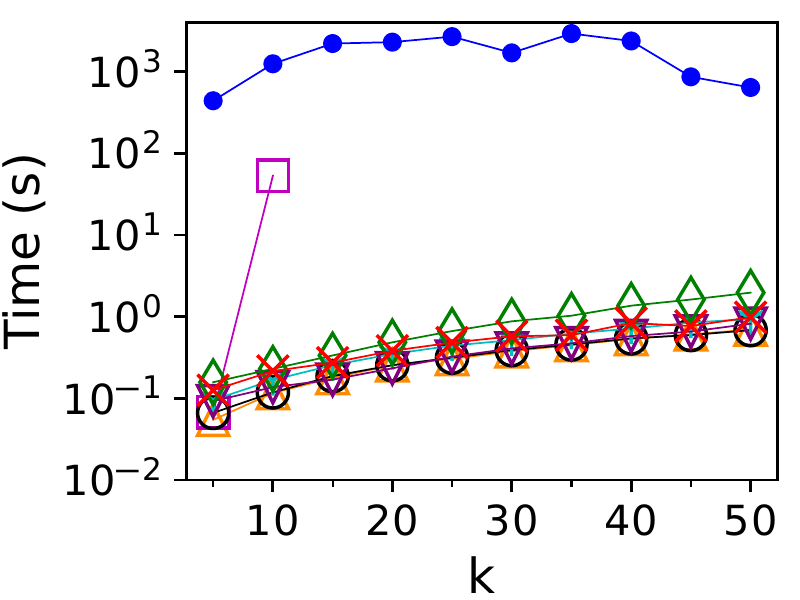}
  }
  \hfill
  \subcaptionbox{Census (Age)}[.195\linewidth]{
    \includegraphics[height=1in]{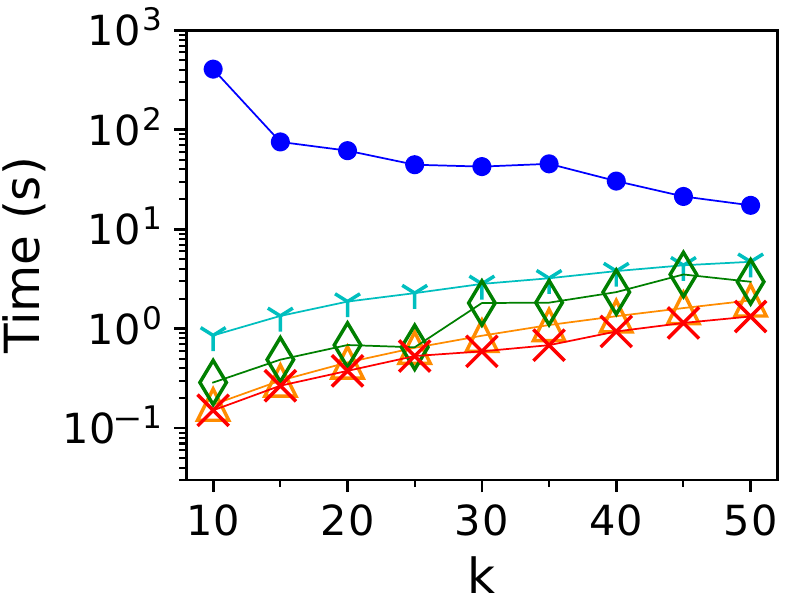}
  }
  \hfill
  \subcaptionbox{Census (S+A)}[.195\linewidth]{
    \includegraphics[height=1in]{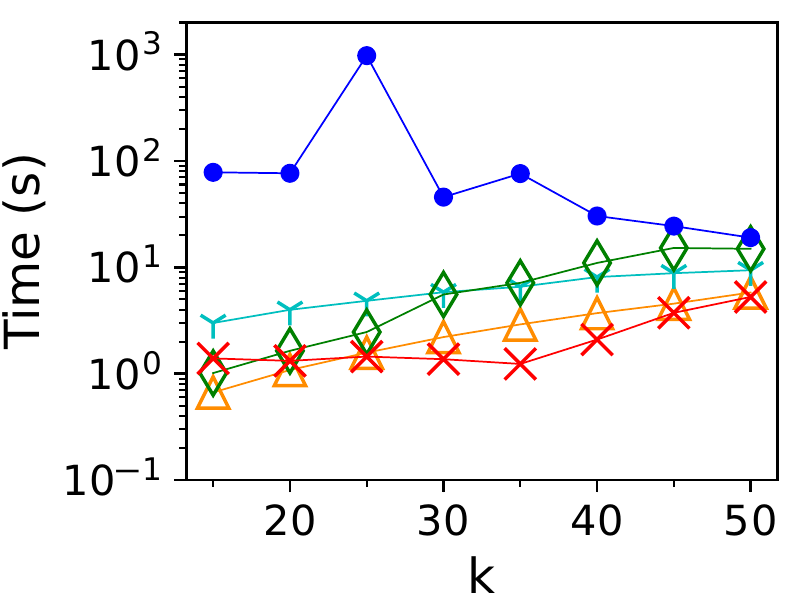}
  }
  \hfill
  \subcaptionbox{Twitter (Sex)}[.195\linewidth]{
    \includegraphics[height=1in]{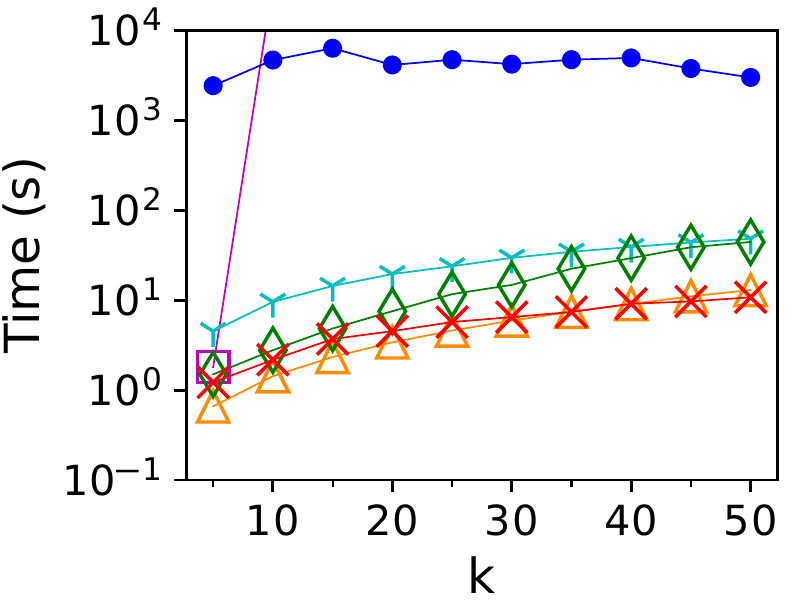}
  }
  \caption{Running time of different algorithms with varying solution size $k$ on small datasets.}
  \label{fig:exp:time}
\end{figure}

\begin{figure}[t]
  \centering
  \includegraphics[width=0.75\textwidth]{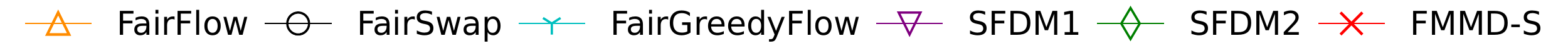}\\
  \subcaptionbox{Adult (Sex)}[.195\linewidth]{
    \includegraphics[height=1in]{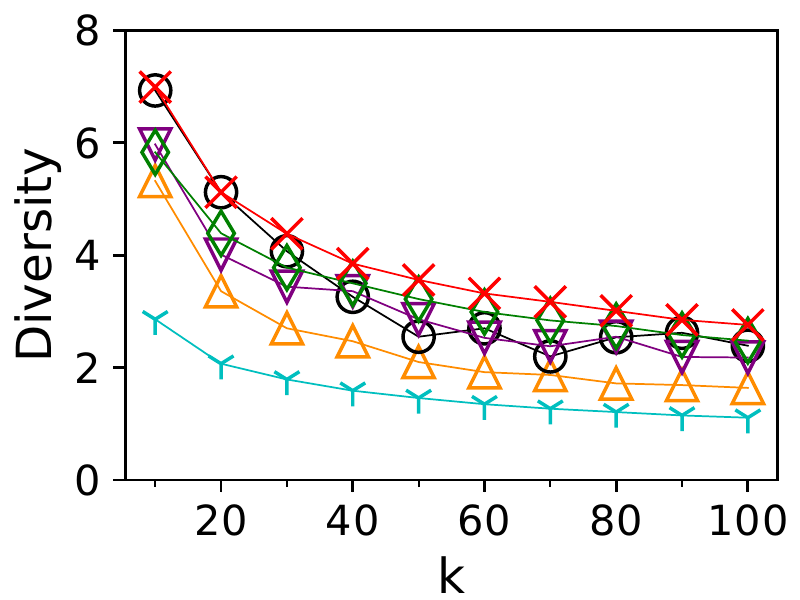}
  }
  \hfill
  \subcaptionbox{Adult (Race)}[.195\linewidth]{
    \includegraphics[height=1in]{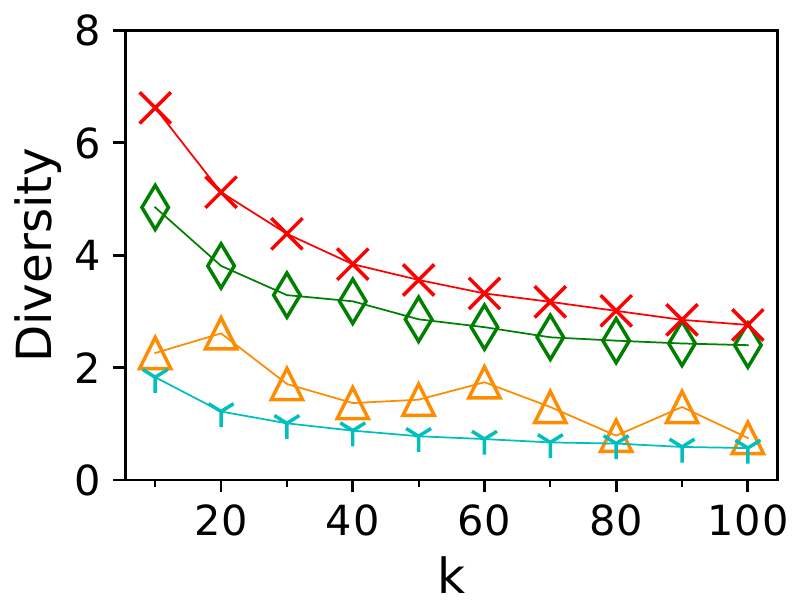}
  }
  \hfill
  \subcaptionbox{Adult (S+R)}[.195\linewidth]{
    \includegraphics[height=1in]{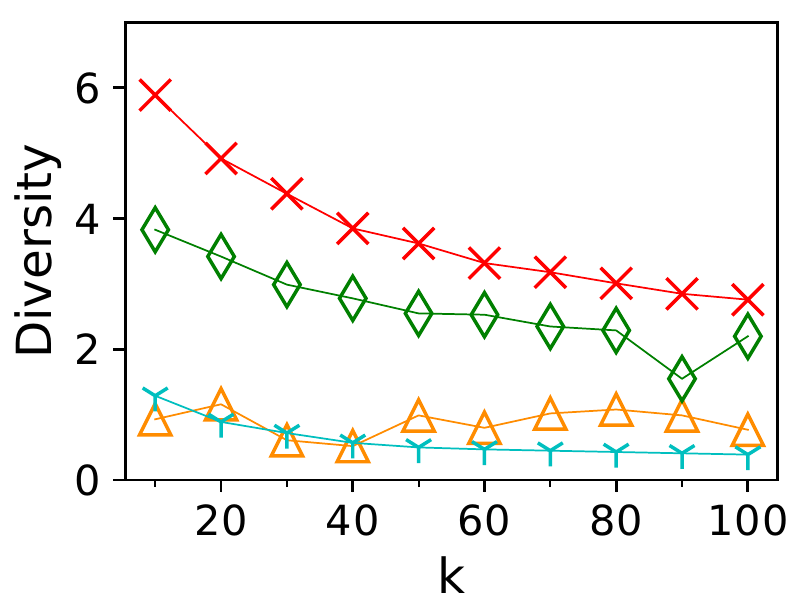}
  }
  \hfill
  \subcaptionbox{CelebA (Sex)}[.195\linewidth]{
    \includegraphics[height=1in]{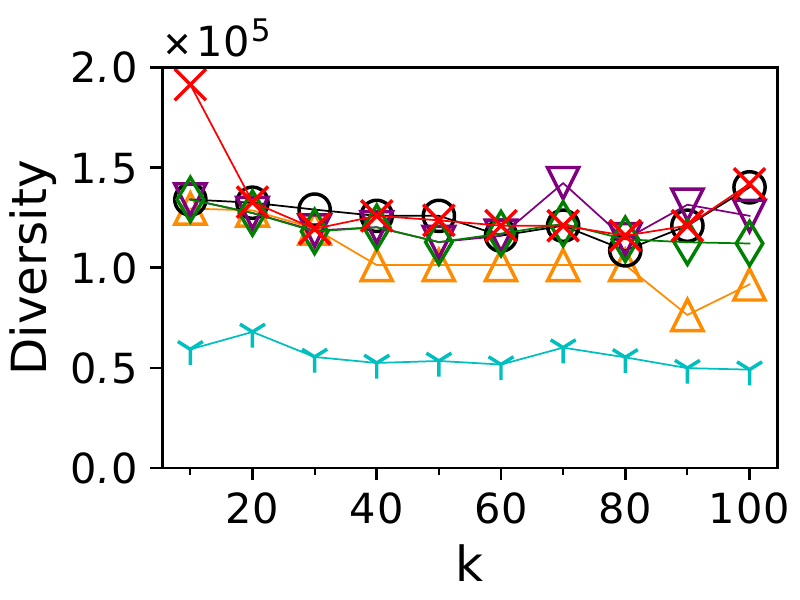}
  }
  \hfill
  \subcaptionbox{CelebA (Age)}[.195\linewidth]{
    \includegraphics[height=1in]{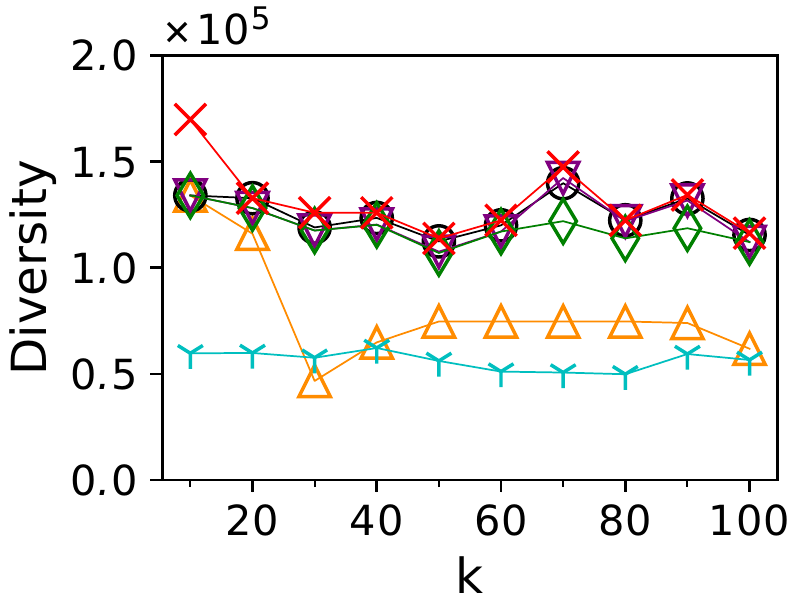}
  }
  \\
  \subcaptionbox{CelebA (S+A)}[.195\linewidth]{
    \includegraphics[height=1in]{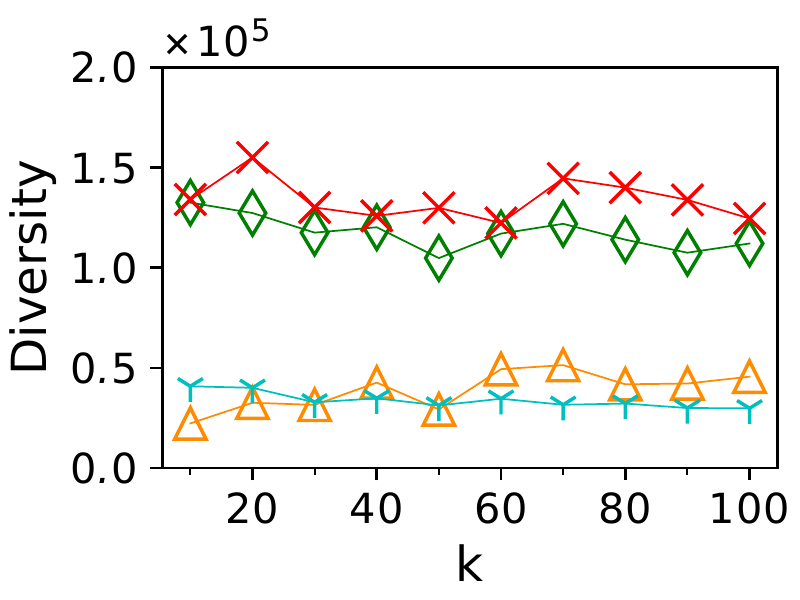}
  }
  \hfill
  \subcaptionbox{Census (Sex)}[.195\linewidth]{
    \includegraphics[height=1in]{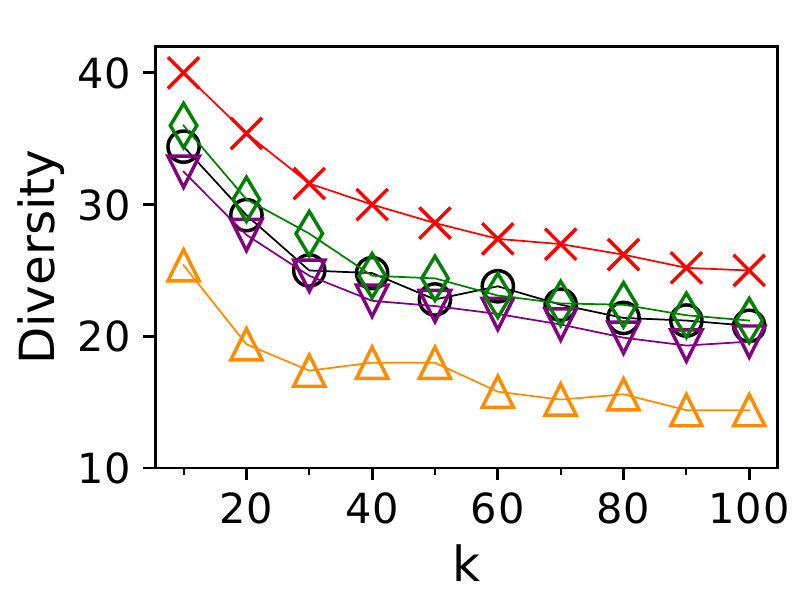}
  }
  \hfill
  \subcaptionbox{Census (Age)}[.195\linewidth]{
    \includegraphics[height=1in]{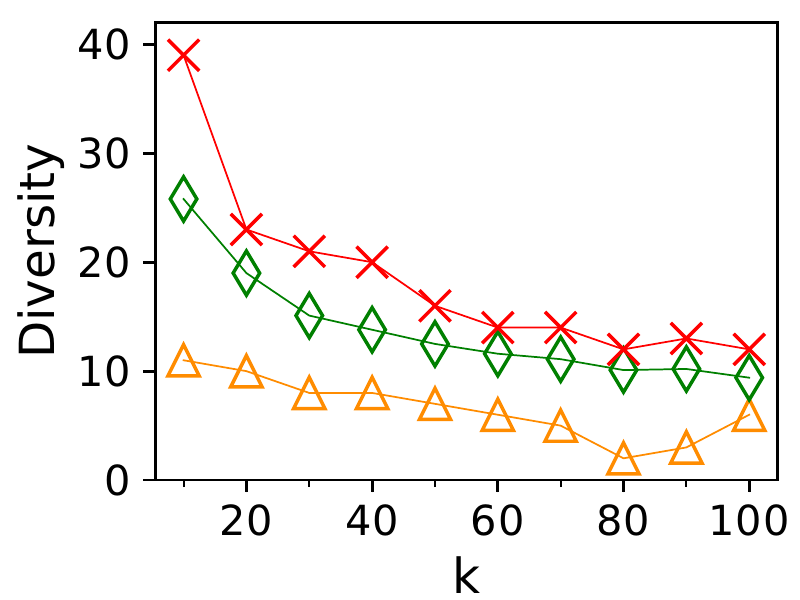}
  }
  \hfill
  \subcaptionbox{Census (S+A)}[.195\linewidth]{
    \includegraphics[height=1in]{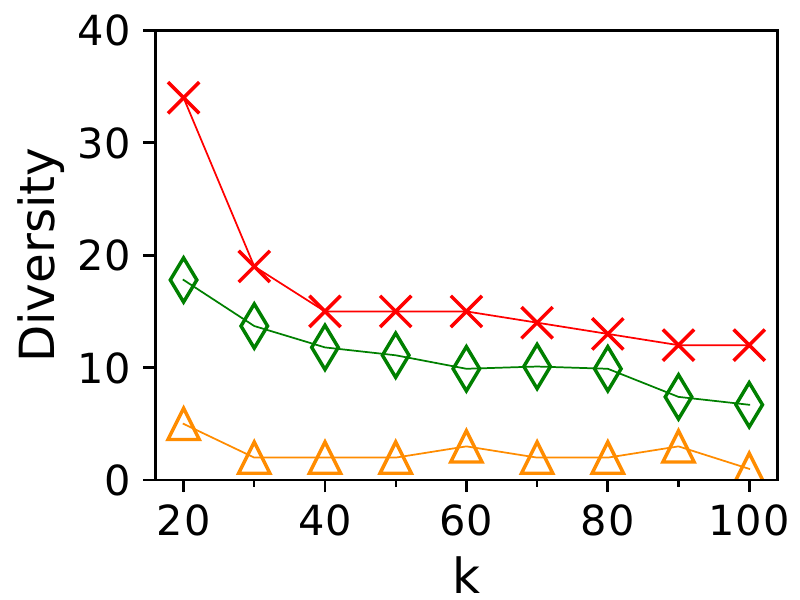}
  }
  \hfill
  \subcaptionbox{Twitter (Sex)}[.195\linewidth]{
    \includegraphics[height=1in]{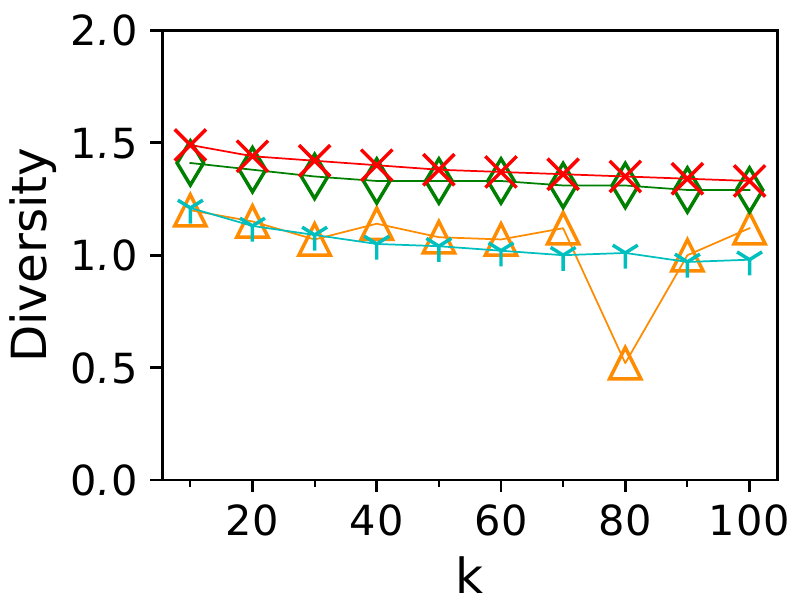}
  }
  \caption{Diversity values of the solutions of different algorithms with varying solution size $k$ on full datasets.}
  \label{fig:exp:full:div}
\end{figure}

\begin{figure}[t]
  \centering
  \includegraphics[width=0.75\textwidth]{figure/full-k/legend.pdf}\\
  \subcaptionbox{Adult (Sex)}[.195\linewidth]{
    \includegraphics[height=1in]{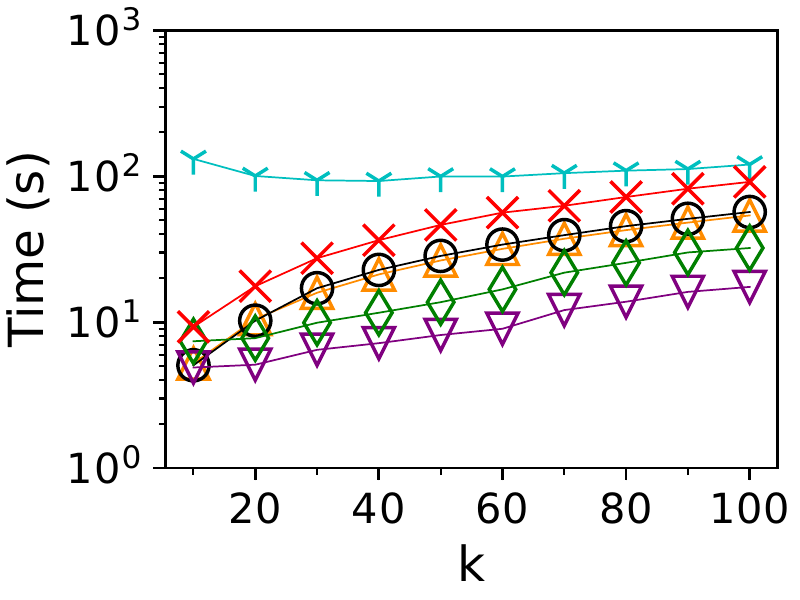}
  }
  \hfill
  \subcaptionbox{Adult (Race)}[.195\linewidth]{
    \includegraphics[height=1in]{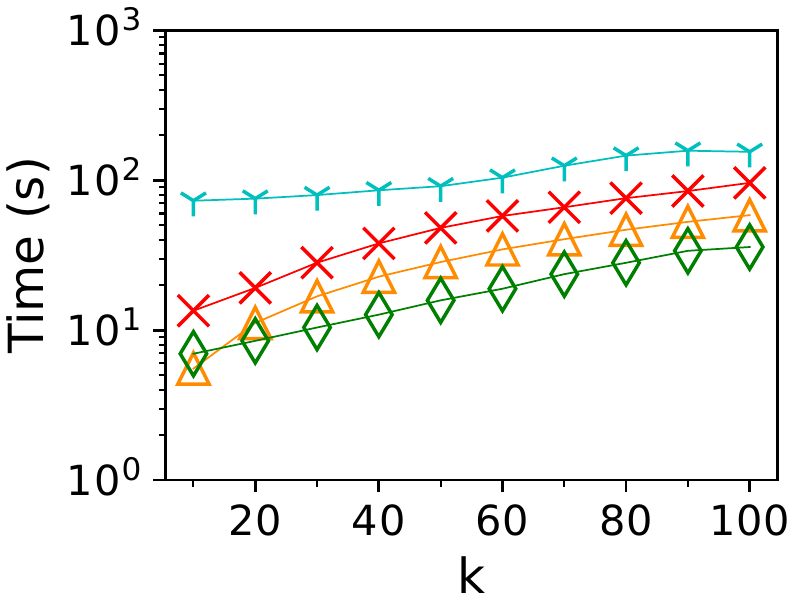}
  }
  \hfill
  \subcaptionbox{Adult (S+R)}[.195\linewidth]{
    \includegraphics[height=1in]{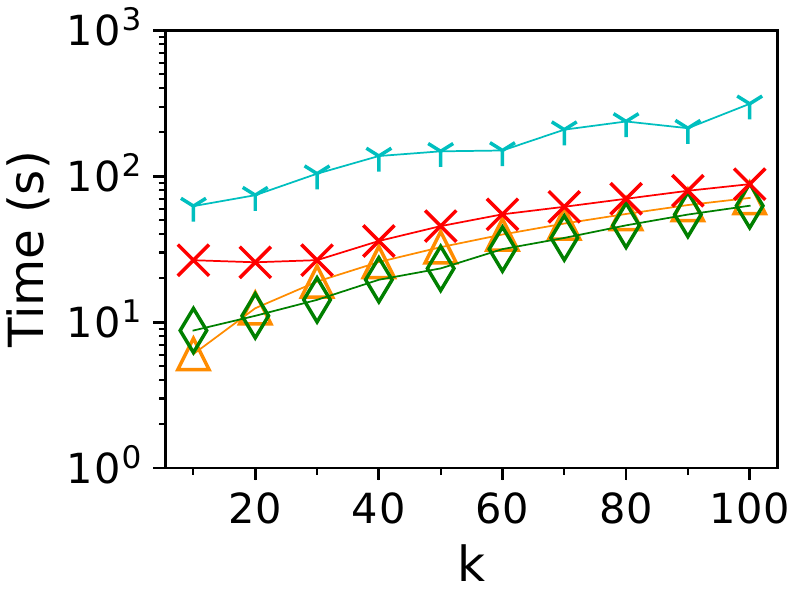}
  }
  \hfill
  \subcaptionbox{CelebA (Sex)}[.195\linewidth]{
    \includegraphics[height=1in]{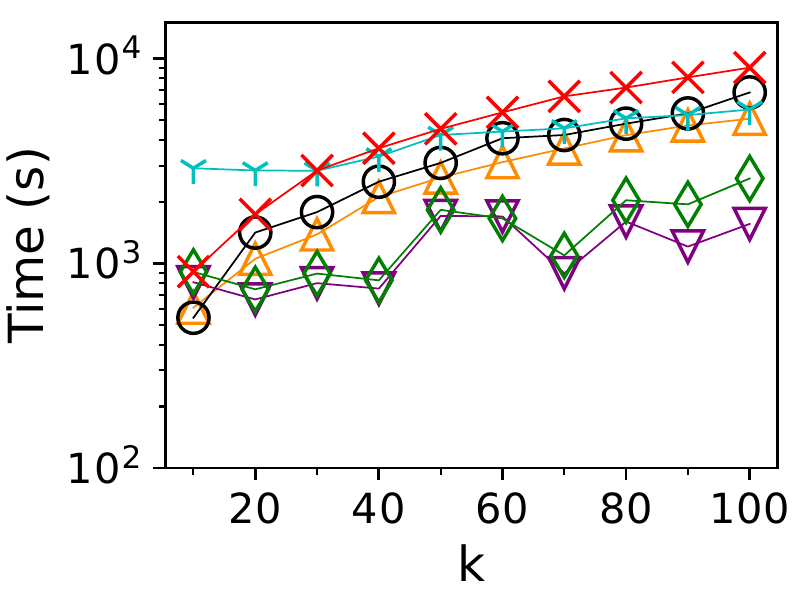}
  }
  \hfill
  \subcaptionbox{CelebA (Age)}[.195\linewidth]{
    \includegraphics[height=1in]{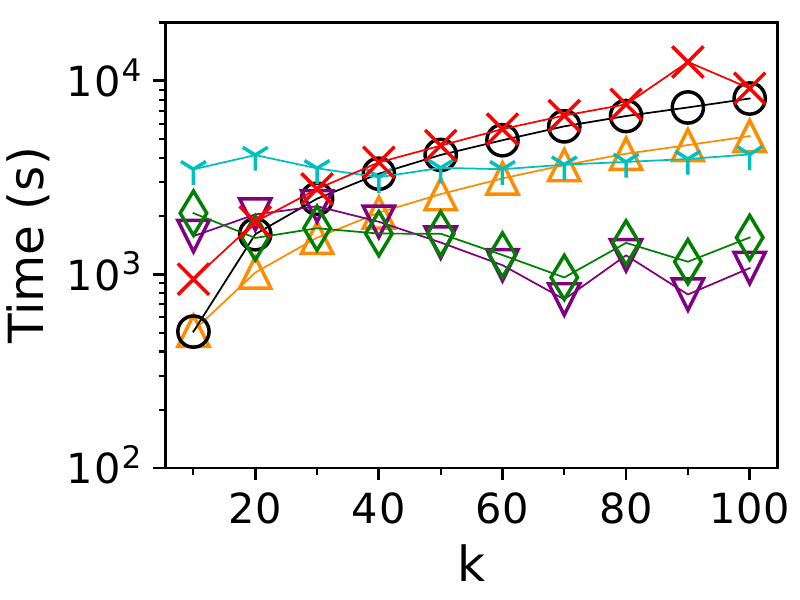}
  }
  \\
  \subcaptionbox{CelebA (S+A)}[.195\linewidth]{
    \includegraphics[height=1in]{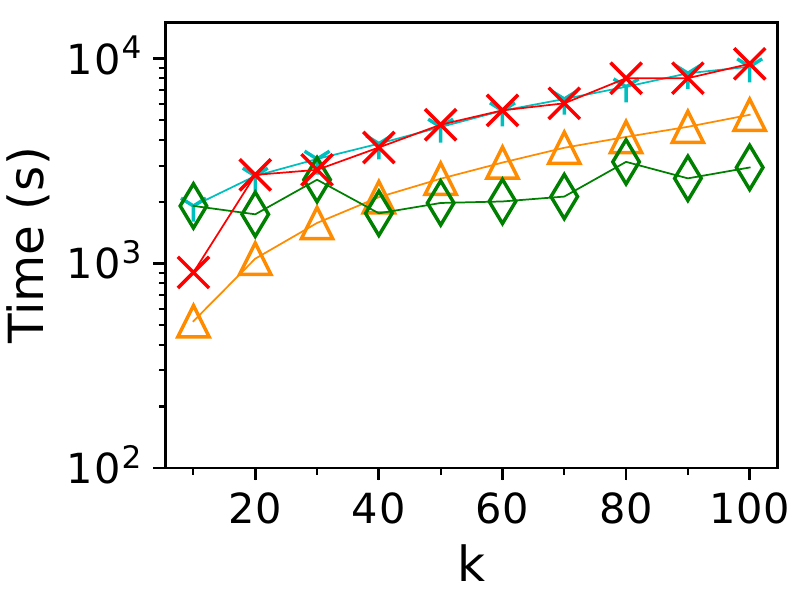}
  }
  \hfill
  \subcaptionbox{Census (Sex)}[.195\linewidth]{
    \includegraphics[height=1in]{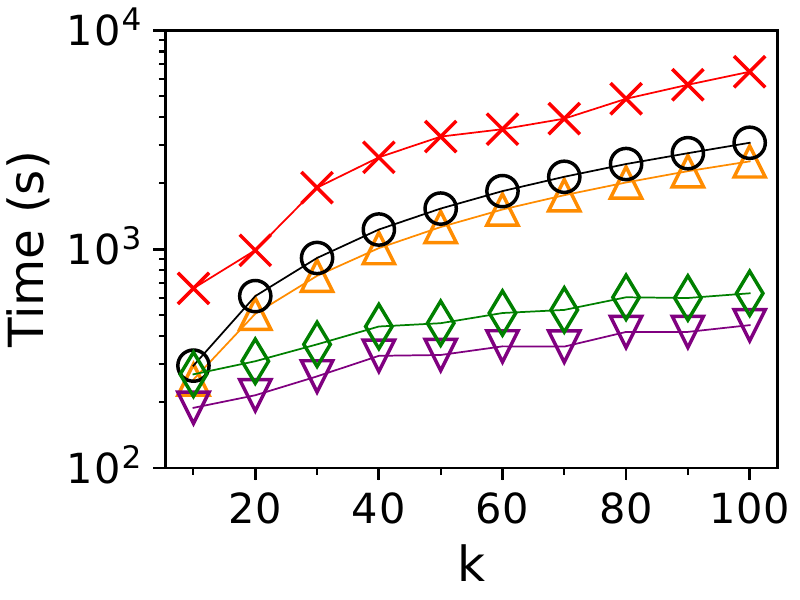}
  }
  \hfill
  \subcaptionbox{Census (Age)}[.195\linewidth]{
    \includegraphics[height=1in]{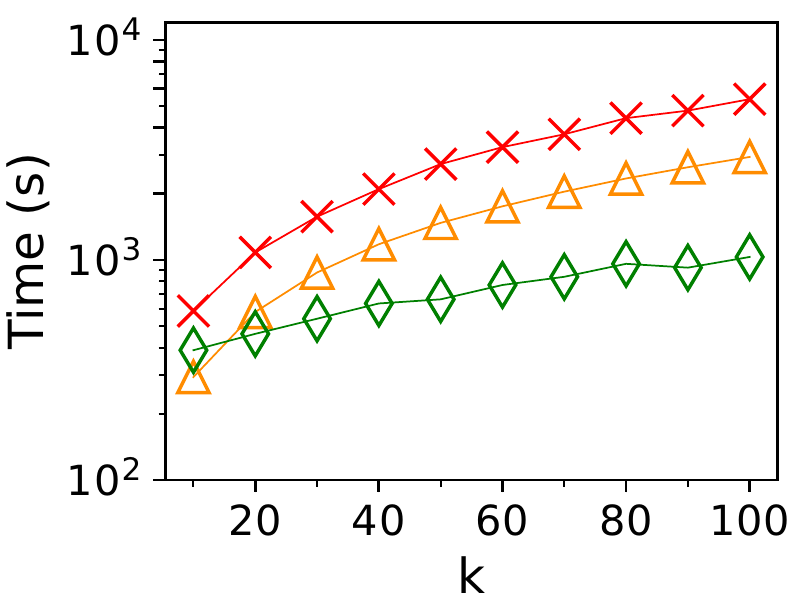}
  }
  \hfill
  \subcaptionbox{Census (S+A)}[.195\linewidth]{
    \includegraphics[height=1in]{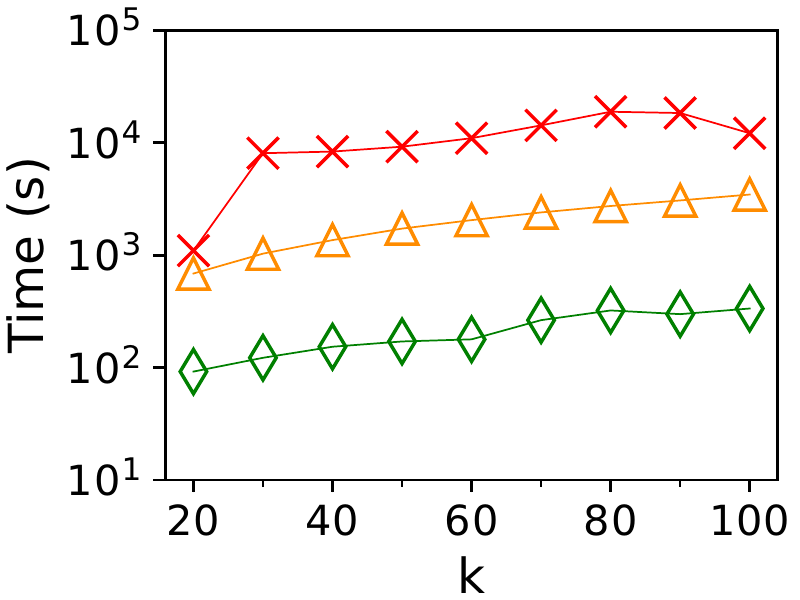}
  }
  \hfill
  \subcaptionbox{Twitter (Sex)}[.195\linewidth]{
    \includegraphics[height=1in]{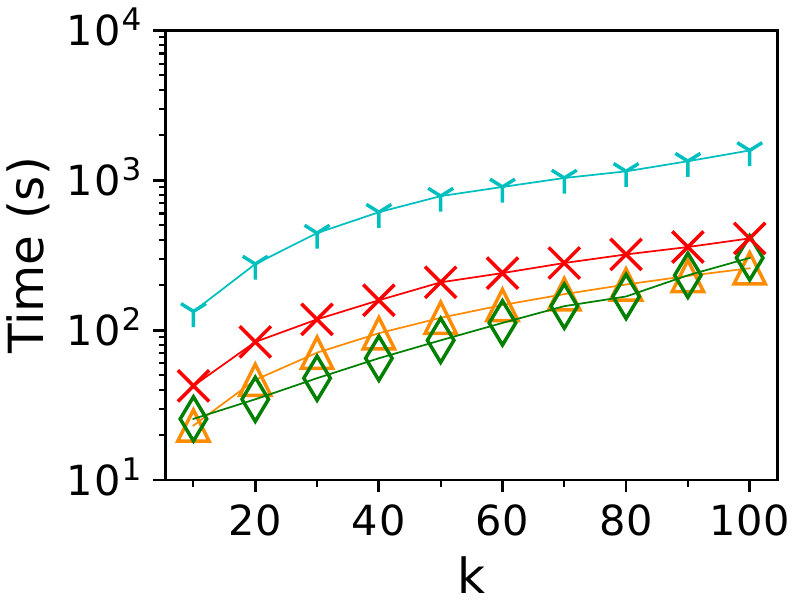}
  }
  \caption{Running time of different algorithms with varying solution size $k$ on full datasets.}
  \label{fig:exp:full:time}
\end{figure}

\subsection{Experimental Results}

Table~\ref{tab:res:small} shows the diversity achieved by different algorithms on ``small'' datasets, i.e., datasets that were obtained by sampling 1,000 items uniformly at random from each full dataset.
Figures~\ref{fig:exp:div}--\ref{fig:exp:time} illustrate the performance of different algorithms on small datasets with varying $k$.
Note that \textsf{FairSwap} and \textsf{SFDM1} are specific for the case of $C = 2$ and \textsf{FairGMM} fails to finish within one day when $C > 3$ or $k > 10$ since it has to enumerate $\binom{kC}{k}$ sets for solution computation.
They are ignored in subsequent tables and figures when they cannot provide valid solutions.

In general, \textsf{FMMD-E} always provides optimal solutions for FMMD within the time limit (i.e., 24 hours) when $n = 1,000$.
The \emph{price of fairness}, measured by the decrease in diversity due to the fairness constraints, is marginal on all datasets except \emph{Adult} with $C \geq 5$.
\textsf{FMMD-S} shows much higher solution quality (up to $3.9\times$ greater in diversity value) than all approximation algorithms except \textsf{FairGMM}.
Although \textsf{FairGMM} sometimes provides slightly better solutions than \textsf{FMMD-S}, it runs more than two orders of magnitude slower.
Moreover, the diversity values of all algorithms drop with $k$ because $div$ is a monotonically non-increasing function.
The running time is independent of $k$ for \textsf{FMMD-E}, grows exponentially with $k$ for \textsf{FairGMM}, and increases linearly with $k$ for \textsf{FMMD-S} and other algorithms, which all follow from their time complexities.

\begin{table}[tb]
\scriptsize
\setlength\tabcolsep{2pt}
\centering
\caption{Results of different algorithms on full datasets for solution size $k=50$. \textsf{FairGMM} and \textsf{FMMD-E} are omitted because they cannot provide any solution within the time limit (i.e., 24 hours).}
\label{tab:res:full}
\begin{tabular}{ccrrrrrrrrrrrr}
\toprule
\multirow{2}{*}{\textbf{Dataset}} & \multirow{2}{*}{\textbf{Group}} & \multicolumn{2}{c}{\textsf{FairSwap}} & \multicolumn{2}{c}{\textsf{FairFlow}} & \multicolumn{2}{c}{\textsf{FairGreedyFlow}} & \multicolumn{2}{c}{\textsf{SFDM1}} & \multicolumn{2}{c}{\textsf{SFDM2}} & \multicolumn{2}{c}{\textsf{FMMD-S}} \\
& & \textbf{diversity} & \textbf{time(s)} & \textbf{diversity} & \textbf{time(s)} & \textbf{diversity} & \textbf{time(s)} & \textbf{diversity} & \textbf{time(s)} & \textbf{diversity} & \textbf{time(s)} & \textbf{diversity} & \textbf{time(s)} \\
\midrule 
\multirow{3}{*}{Adult}  & Sex  & 2.55          & 28.45   & 2.10     & 26.47  & 1.46    & 99.61         & 2.86     & 8.17         & 3.22     & 13.71 & \textbf{3.56}     & 46.38  \\
                        & Race & \multicolumn{2}{c}{N/A} & 1.43     & 28.47  & 0.78    & 91.24         & \multicolumn{2}{c}{N/A} & 2.86     & 15.81 & \textbf{3.56}     & 48.10  \\
                        & S+R  & \multicolumn{2}{c}{N/A} & 0.99     & 32.64  & 0.50    & 148.12        & \multicolumn{2}{c}{N/A} & 2.55     & 23.35 & \textbf{3.61}     & 45.51  \\
\midrule
\multirow{3}{*}{CelebA} & Sex  & 125865.8      & 3093.9  & 101303.9 & 2624.1 & 57696.8 & 3558.0        & \textbf{128450.2} & 1473.4       & 117342.5 & 1354.4 & 123639.0 & 4536.4 \\
                        & Age  & 112387.6      & 4141.2  & 74679.4  & 2598.8 & 55470.0 & 2260.8        & \textbf{126446.9} & 1024.0       & 121056.9 & 1222.3 & 113798.5 & 4626.6 \\
                        & S+A  & \multicolumn{2}{c}{N/A} & 29278.4  & 2589.5 & 34066.4 & 3871.3        & \multicolumn{2}{c}{N/A} & 118598.7 & 1141.7 & \textbf{129866.3} & 4752.5 \\
\midrule
\multirow{3}{*}{Census} & Sex  & 22.8          & 1533.3  & 16.8     & 1250.4 & \multicolumn{2}{c}{N/A} & 22.3     & 328.7        & 24.4     & 459.14 & \textbf{28.6}     & 3268.8 \\
                        & Age  & \multicolumn{2}{c}{N/A} & 5.2      & 1461.5 & \multicolumn{2}{c}{N/A} & \multicolumn{2}{c}{N/A} & 12.5     & 662.76 & \textbf{16.0}     & 2728.2 \\
                        & S+A  & \multicolumn{2}{c}{N/A} & 3.4      & 1723.5 & \multicolumn{2}{c}{N/A} & \multicolumn{2}{c}{N/A} & 11.1     & 170.98 & \textbf{15.0}     & 9224.8 \\
\midrule
Twitter                 & Sex  & \multicolumn{2}{c}{N/A} & 1.08     & 120.85 & 1.04    & 783.7         & \multicolumn{2}{c}{N/A} & 1.33     & 85.56  & \textbf{1.38}     & 208.52 \\
\bottomrule
\end{tabular}
\end{table}

\begin{figure}[ht]
  \centering
  \includegraphics[width=0.9\textwidth]{figure/small-k/legend.pdf}
  \\
  \subcaptionbox{Synthetic ($n = 1,000$)}[.32\linewidth]{\centering\includegraphics[height=0.8in]{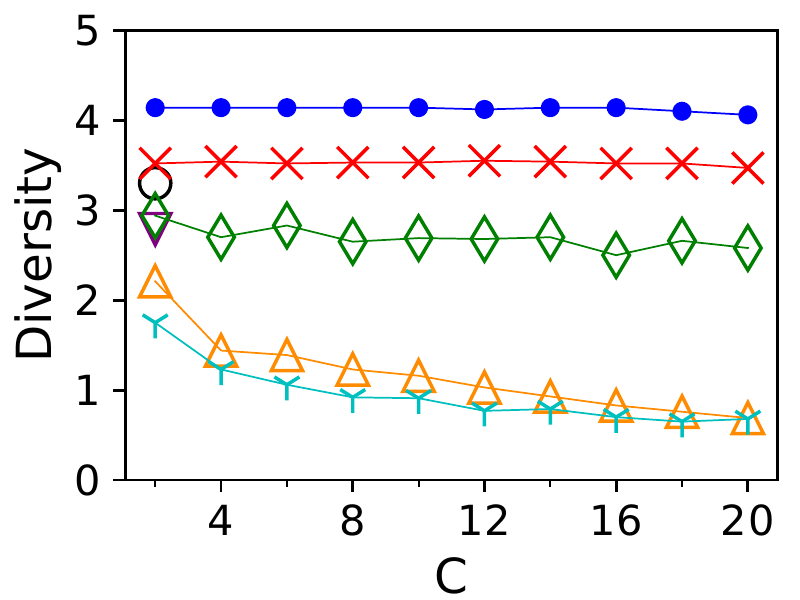}\includegraphics[height=0.8in]{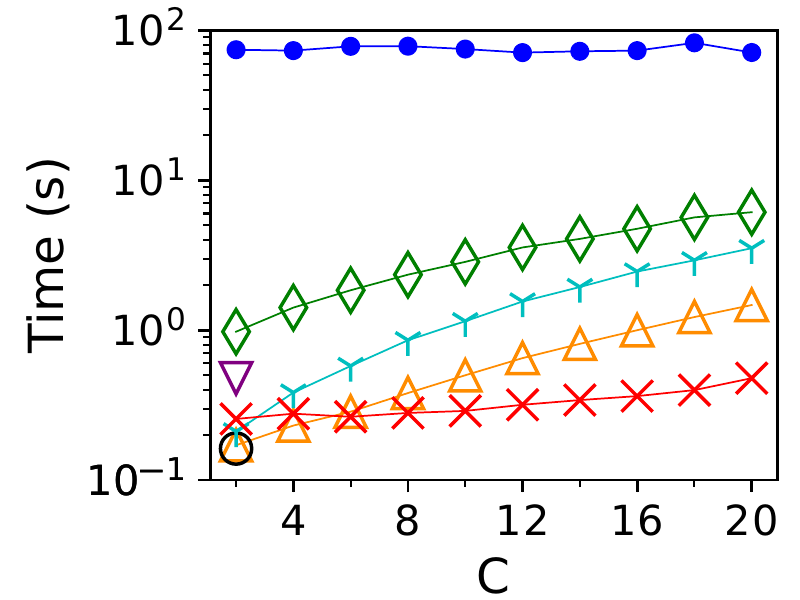}}
  \subcaptionbox{Synthetic ($C = 2$)}[.32\linewidth]{\centering\includegraphics[height=0.8in]{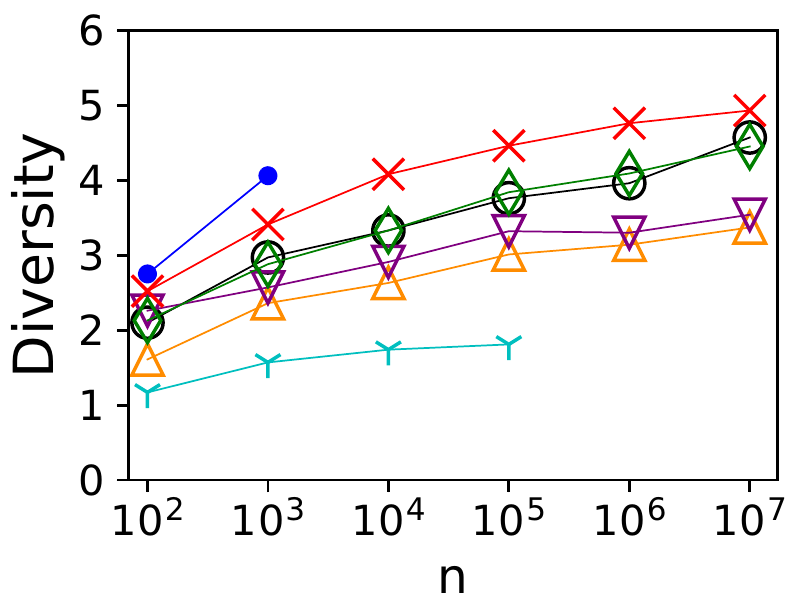}\includegraphics[height=0.8in]{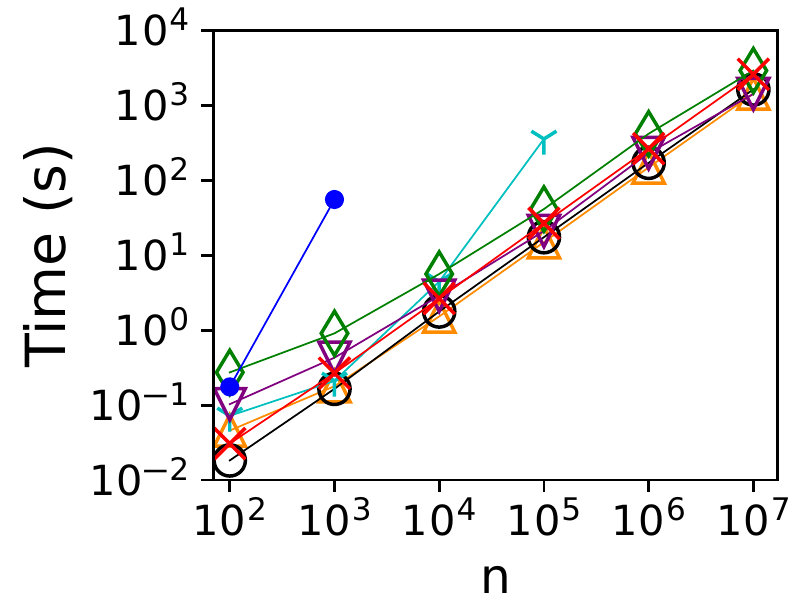}}
  \subcaptionbox{Synthetic ($C = 10$)}[.32\linewidth]{\centering\includegraphics[height=0.8in]{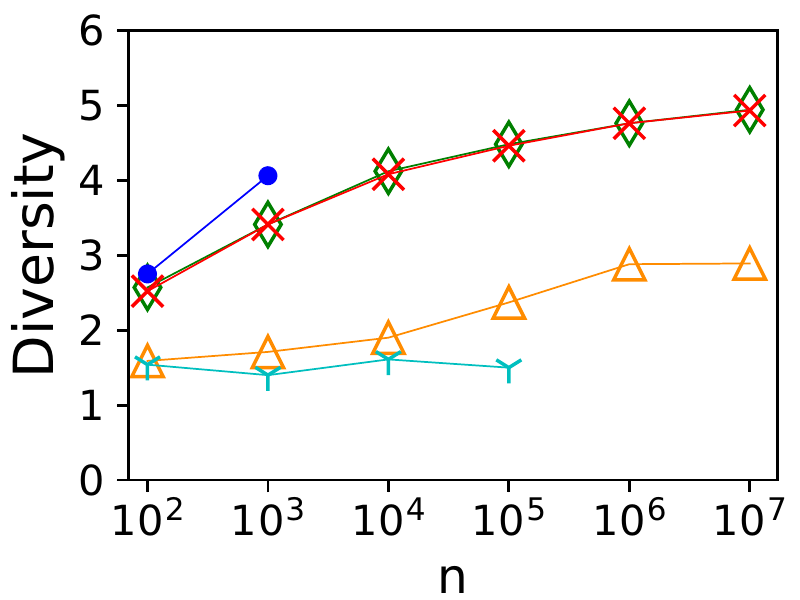}\includegraphics[height=0.8in]{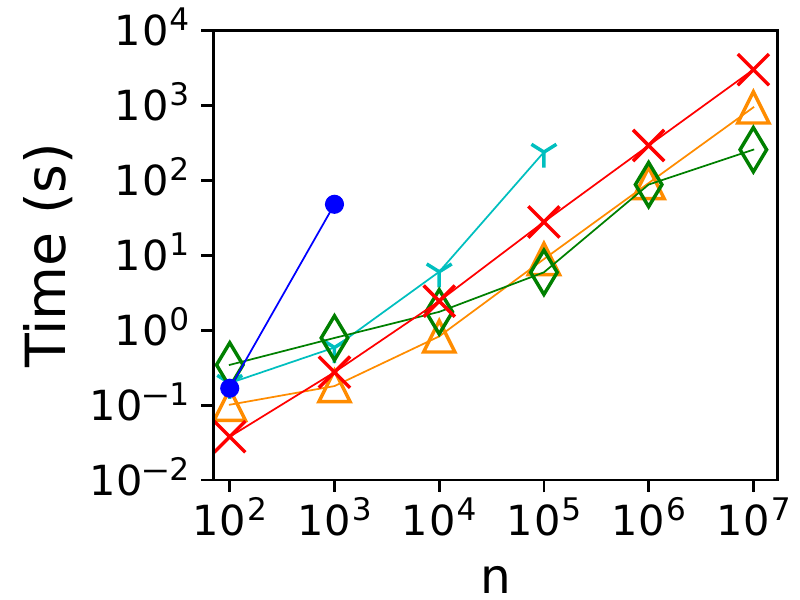}}
  \caption{Results of different algorithms by varying the number $C$ of groups and dataset size $n$ for solution size $k = 20$. Since \textsf{FMMD-E} does not provide any solution within the time limit on datasets with $n > 1,000$, \textsf{FairGMM} does not provide any solution within the time limit for $k = 20$, and \textsf{FairSwap} and \textsf{SFDM1} cannot work with $C > 2$, they are not plotted in these cases.}
  \label{fig:syn}
\end{figure}

Table~\ref{tab:res:full} presents the diversity values and running time of different algorithms for solution size $k=50$ on all the full datasets.
The performance of different algorithms by varying the solution size $k$ from $10$ to $100$ on full datasets (with all the items) is presented in Figures~\ref{fig:exp:full:div}--\ref{fig:exp:full:time}.
In general, the running time of all algorithms increases substantially with $n$ and $dim$. \textsf{FairGMM} and \textsf{FMMD-E} fail to finish within one day and thus are omitted from Table~\ref{tab:res:full}.
\textsf{FMMD-S} provides better solutions (up to $7.2\times$ higher in diversity value) than all the baselines in most cases.
The only exception is that \textsf{FMMD-S} shows slightly lower solution quality than \textsf{FairSwap} and \textsf{SFDM1} on \emph{CelebA} when $C = 2$.
This is because of the extremely high dimensionality of \emph{CelebA} ($d = 25,088$), where the distances between different pairs of points are less distinguishable.
In such cases, the thresholding method in \textsf{FMMD-S} is inferior to the local search methods in \textsf{FairSwap} and \textsf{SFDM1}.
Moreover, the time efficiency of \textsf{FMMD-S} is lower than the baselines, as solving ILPs is often time-consuming.
Nevertheless, on \emph{Census}, i.e., the largest dataset with more than two million items, \textsf{FMMD-S} still finishes the computation within 3 hours.

To evaluate the scalability of different algorithms, we vary the number $C$ of groups and the number $n$ of points on synthetic datasets.
In particular, each dataset consists of ten two-dimensional Gaussian isotropic blobs with random centers in $[-10, 10]^2$ and identity covariance matrices.
Each point is assigned to one of the $C$ groups uniformly at random.
The Euclidean distance is used as the distance metric.
For fixed $C = 2$ or $10$, we obtain six datasets with $n = 10^2, 10^3, \ldots, 10^7$; and for fixed $n = 1,000$, ten datasets with $C = 2, 4, \ldots, 20$.

The performance of different algorithms by varying $n$ and $C$ on synthetic datasets for solution size $k = 20$ is presented in Figure~\ref{fig:syn}.
In terms of solution quality, the diversity values are steady for \textsf{FMMD-E} and \textsf{FMMD-S} but significantly drop for all other algorithms when $C$ increases.
In terms of efficiency, the running time of \textsf{FMMD-E} is hardly affected by $C$.
Other algorithms run slower when $C$ is larger.
Nevertheless, \textsf{FMMD-S} runs faster than any other algorithm when $C \geq 8$, and its advantages in time efficiency become more significant with increasing $C$.
Finally, all algorithms' diversity values and running time grow with the dataset size $n$.
\textsf{FMMD-E} cannot scale to large datasets due to its exponential time complexity.
All in all, \textsf{FMMD-S} outperforms all the other approximation algorithms in terms of solution quality for different $C$ or $n$, and its advantages become more apparent when $C$ or $n$ is larger.
These results confirm the scalability of \textsf{FMMD-S} concerning the group size $C$ and dataset size $n$.

\section{Conclusion}
\label{sec:conclusion}

We investigated the problem of max-min diversification with fairness constraints (FMMD) in this paper. We proposed an exact ILP-based algorithm for this problem on small datasets. We further designed a scalable $\frac{1-\varepsilon}{5}$-approximation algorithm, where $\varepsilon \in (0,1)$, on massive datasets based on our exact algorithm and the notion of \emph{coresets}. Extensive experimental results on four real-world datasets confirmed the effectiveness, efficiency, and scalability of our proposed algorithms.

While a step forward in both theoretical and experimental aspects of the algorithms for the fair variant of diversity maximization, our work leaves many open problems for future exploration.
A natural question is whether there is any polynomial-time $O(1)$-approximation algorithm for FMMD, since our \textsf{FMMD-S} algorithm has an approximation factor of $\frac{1-\varepsilon}{5}$ but runs in polynomial time only when $C=O(1)$ and $k=o(\log{n})$, whereas the best-known polynomial-time algorithm in~\cite{DBLP:conf/icdt/Addanki0MM22} only achieves an approximation factor of $\frac{1}{C+1}$.
Moreover, it would also be interesting to study the fairness-aware variants of other diversity measures (e.g., the ones in~\cite{DBLP:conf/pods/IndykMMM14,DBLP:conf/nips/BhaskaraGMS16}).

\appendix

\section{Data Preparation}
\label{sec:datasets}

Detailed information about the four real-world datasets we use and the preprocessing procedures for them are presented as follows.
\begin{itemize}[nolistsep]
  \item \textbf{Adult} is retrieved from the UCI Machine Learning Repository\footnote{\url{archive.ics.uci.edu/ml/datasets/adult}}. It is a collection of 48,842 records from the 1994 US Census database. We select six numeric attributes as features and normalize each to have zero mean and unit standard deviation. The $l_2$-distance (Euclidean distance) is used as the distance metric. The groups are generated from two demographic attributes: \emph{sex} and \emph{race}. By using them individually and in combination, there are 2 (sex), 5 (race), and 10 (sex + race) groups, respectively.
  \item \textbf{CelebA} is provided by Liu \emph{et al.}~\cite{DBLP:conf/iccv/LiuLWT15} on the website\footnote{\url{mmlab.ie.cuhk.edu.hk/projects/CelebA.html}}. It is a set of 202,599 images of human faces. We get a 25,088-dimensional ($512 \times 7 \times 7$) feature vector for each image from the pre-trained VGG16 model in Keras\footnote{\url{https://keras.io/api/applications/vgg/}}. The $l_1$-distance (Manhattan distance) between feature vectors is used as the distance metric. We generate 2 groups from a human-annotated class label ``\emph{sex}'' \{`female', `male'\}, 2 groups from another human-annotated class label ``\emph{age}'' \{`young', `not young'\}, and 4 groups from both of them, respectively.
  \item \textbf{Census} is also retrieved from the UCI Machine Learning Repository\footnote{\url{archive.ics.uci.edu/ml/datasets/US+Census+Data+(1990)}}. It is a set of 2,426,116 records from the 1990 US Census data. We take 25 (normalized) numeric attributes as features and use the $l_1$-distance (Manhattan distance) as the distance metric. We generate 2, 7, and 14 groups from two demographic attributes \emph{sex}, \emph{age}, and both of them, respectively.
  \item \textbf{Twitter} is retrieved from Kaggle\footnote{\url{www.kaggle.com/crowdflower/twitter-user-gender-classification}}. It is a collection of 18,836 tweets with user profiles. We transform each tweet into a 1,024-dimensional feature vector using the sentence-transformer model\footnote{\url{huggingface.co/models?library=sentence-transformers}}~\cite{DBLP:conf/emnlp/ReimersG19} (``\textsf{bert-large-nli-stsb-mean-tokens}''). The angular distance is used as the distance metric. We generate 3 groups from the attribute ``\emph{sex}'' \{`female', `male', `non-human'\} in user profiles.
\end{itemize}
The proportion of each group on each dataset is provided in Table~\ref{tab:prop}.

\begin{table}[tb]
\centering
\scriptsize
\setlength\tabcolsep{3pt}
\caption{Proportion of each group on each dataset.}
\label{tab:prop}
\begin{tabular}{ccl}
\toprule
\textbf{Dataset}        & \textbf{Group} & \textbf{Proportion of Each Group} \\
\midrule
\multirow{3}{*}{Adult}  & Sex            & `Female': 33.2\%, `Male': 66.8\% \\
                        & Race           & \begin{tabular}[l]{@{}l@{}}`Amer-Indian-Eskimo': 1.0\%, `Asian-Pac-Islander': 3.1\%, `Black': 9.6\%,\\ `Others': 0.8\%, `White': 85.5\%\end{tabular} \\
                        & Sex+Race       & \begin{tabular}[l]{@{}l@{}}`AIE+F': 0.4\%, `API+F': 1.1\%, `B+F': 4.7\%, `O+F': 0.3\%, `W+F': 26.7\%,\\ `AIE+M': 0.6\%, `API+M': 2.0\%, `B+M': 4.9\%, `O+M': 0.5\%, `W+M': 58.8\%\end{tabular} \\
\midrule
\multirow{3}{*}{CelebA} & Sex            & `Female': 58.3\%, `Male': 41.7\% \\
                        & Age            & `Not Young': 22.6\%, `Young': 77.4\% \\
                        & Sex+Age        & `NY+F': 7.3\%, `Y+F': 51.0\%, `NY+M': 15.3\%, `Y+M': 26.4\% \\
\midrule
\multirow{3}{*}{Census} & Sex            & `Female' : 51.6\%, `Male' : 48.4\% \\
                        & Age            & \begin{tabular}[l]{@{}l@{}}`A1 (12-)': 18.2\%, `A2 (13-19)': 10.0\%, `A3 (20-29)': 15.3\%,\\ `A4 (30-39)': 16.7\%, `A5 (40-49)': 12.9\%, `A6 (50-64)': 13.6\%,\\ `A7 (65+)': 13.3\%\end{tabular} \\
                        & Sex+Age        & \begin{tabular}[l]{@{}l@{}}`A1+F': 8.9\%, `A2+F': 4.9\%, `A3+F': 7.7\%, `A4+F': 8.5\%, `A5+F': 6.6\%,\\ `A6+F': 7.2\%, `A7+F': 7.9\%, `A1+M': 9.3\%, `A2+M': 5.1\%, `A3+M': 7.6\%,\\ `A4+M': 8.2\%, `A5+M': 6.3\%, `A6+M': 6.5\%, `A7+M': 5.4\%\end{tabular} \\
\midrule
Twitter                 & Sex            & `Female': 35.6\%, `Male': 32.9\%, `Non-Human': 31.5\% \\
\bottomrule
\end{tabular}
\end{table}
\begin{figure}[tb]
  \centering
  \subcaptionbox{Adult (Sex)}[.325\linewidth]{\includegraphics[width=1.1in]{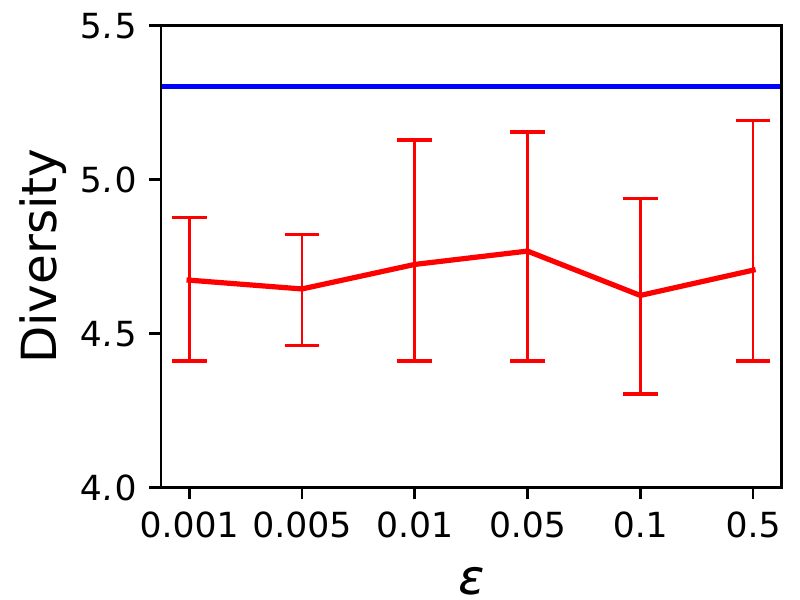}\hfill\includegraphics[width=1.1in]{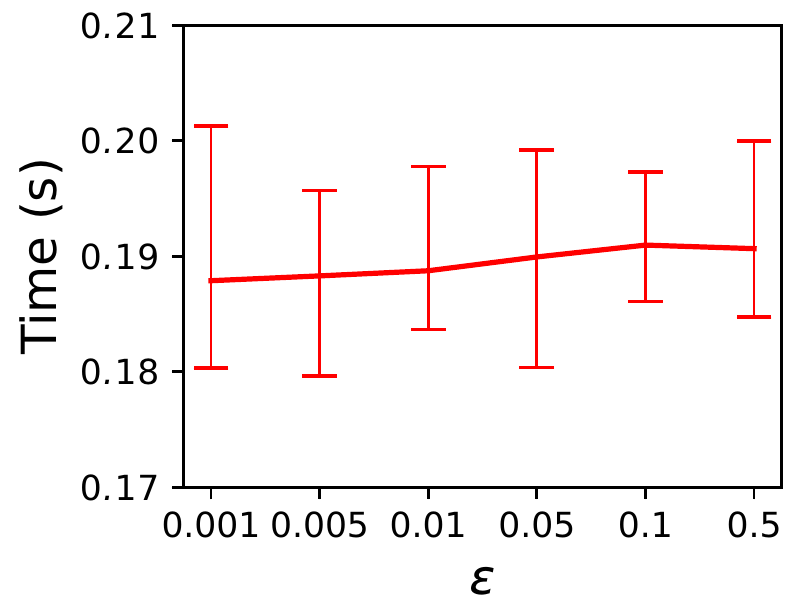}}
  \hfill
  \subcaptionbox{Adult (Race)}[.325\linewidth]{\includegraphics[width=1.1in]{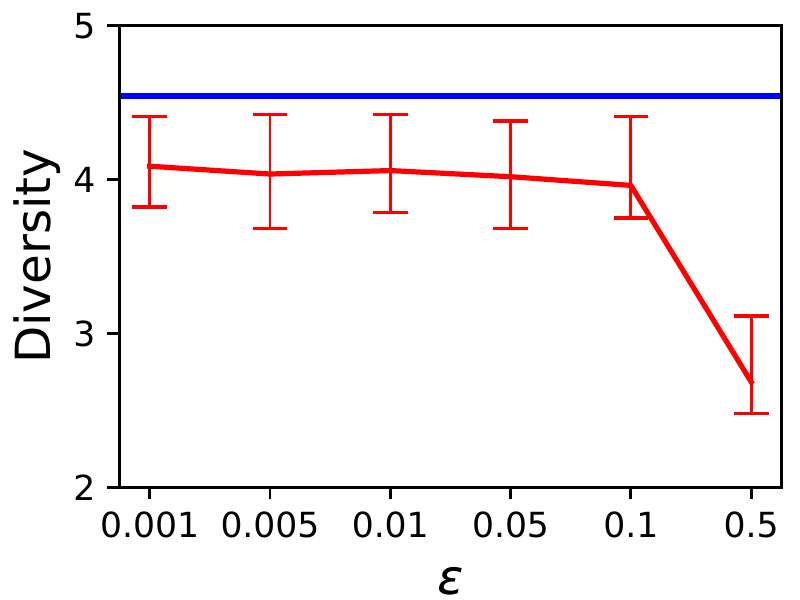}\hfill\includegraphics[width=1.1in]{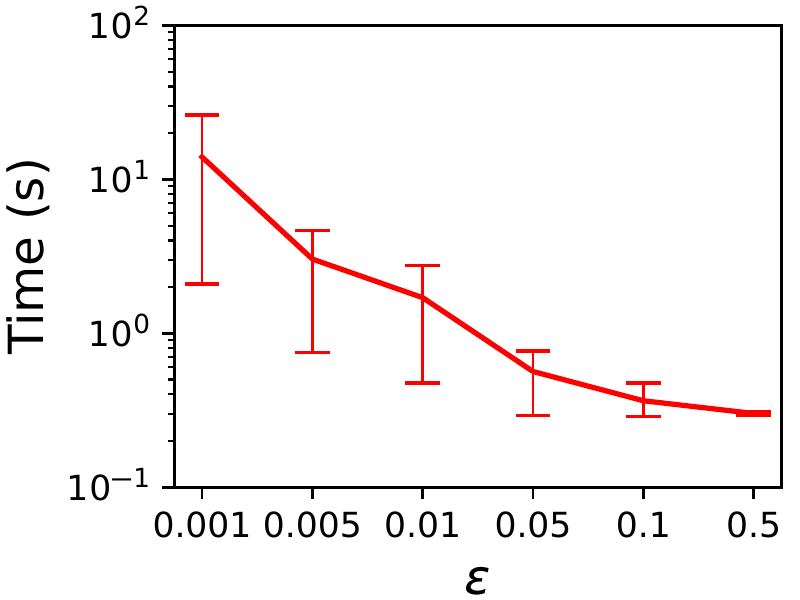}}
  \hfill
  \subcaptionbox{Adult (S+R)}[.325\linewidth]{\includegraphics[width=1.1in]{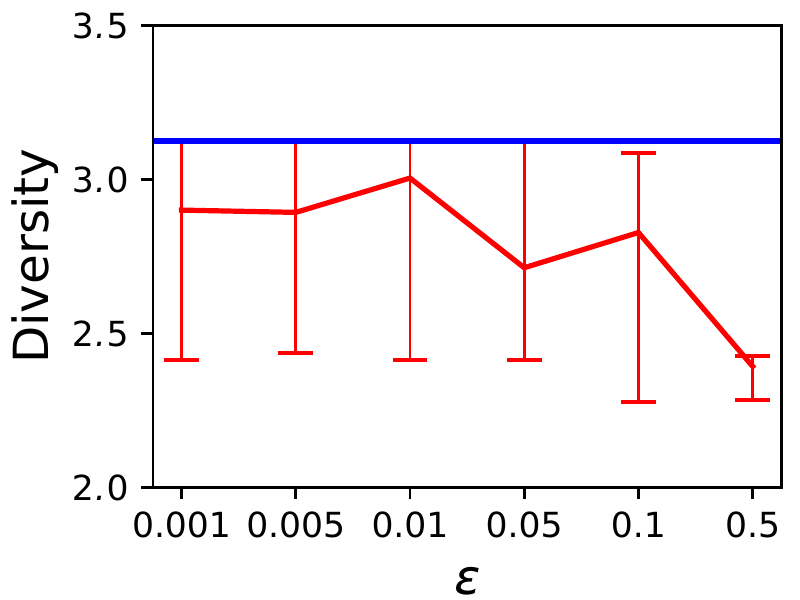}\hfill\includegraphics[width=1.1in]{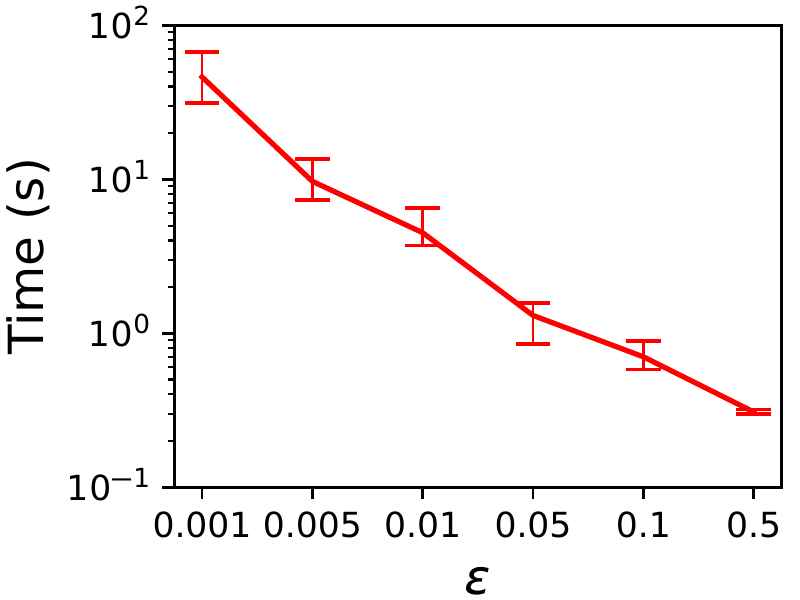}}
  \\
  \subcaptionbox{CelebA (Sex)}[.325\linewidth]{\includegraphics[width=1.1in]{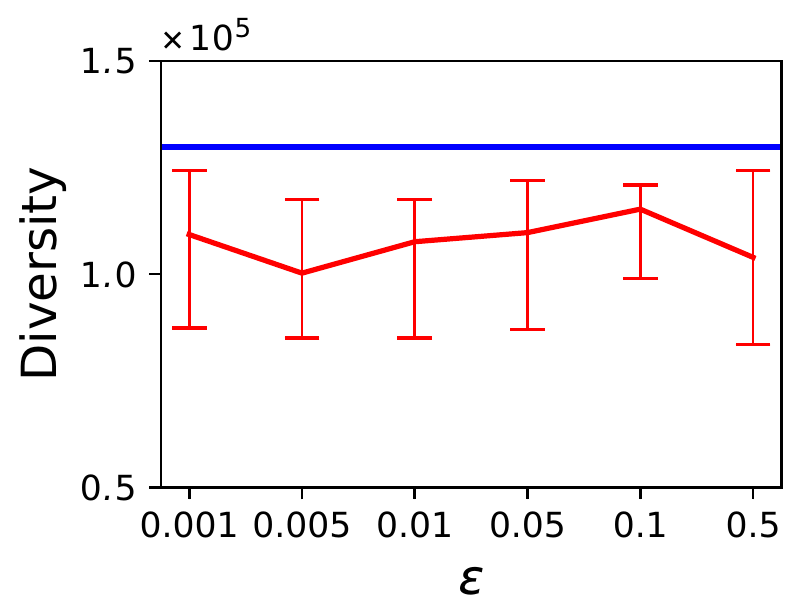}\hfill\includegraphics[width=1.1in]{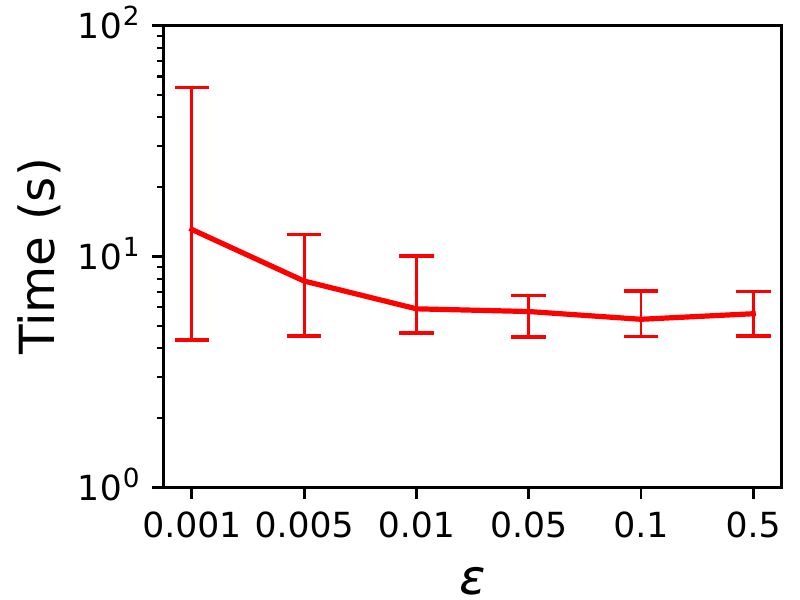}}
  \hfill
  \subcaptionbox{CelebA (Age)}[.325\linewidth]{\includegraphics[width=1.1in]{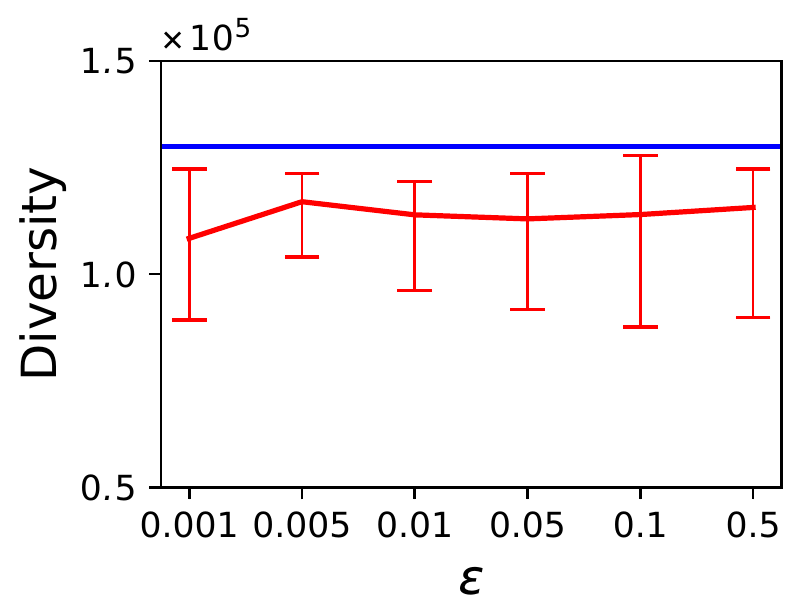}\hfill\includegraphics[width=1.1in]{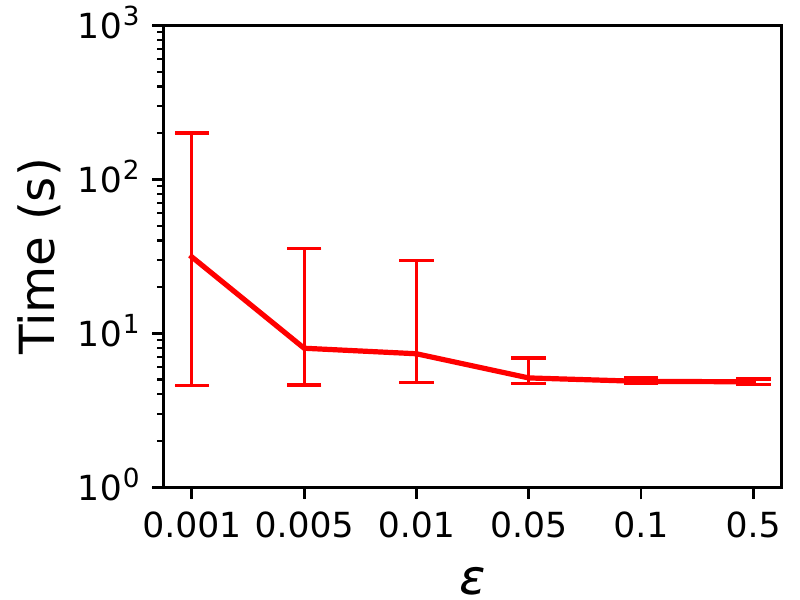}}
  \hfill
  \subcaptionbox{CelebA (S+A)}[.325\linewidth]{\includegraphics[width=1.1in]{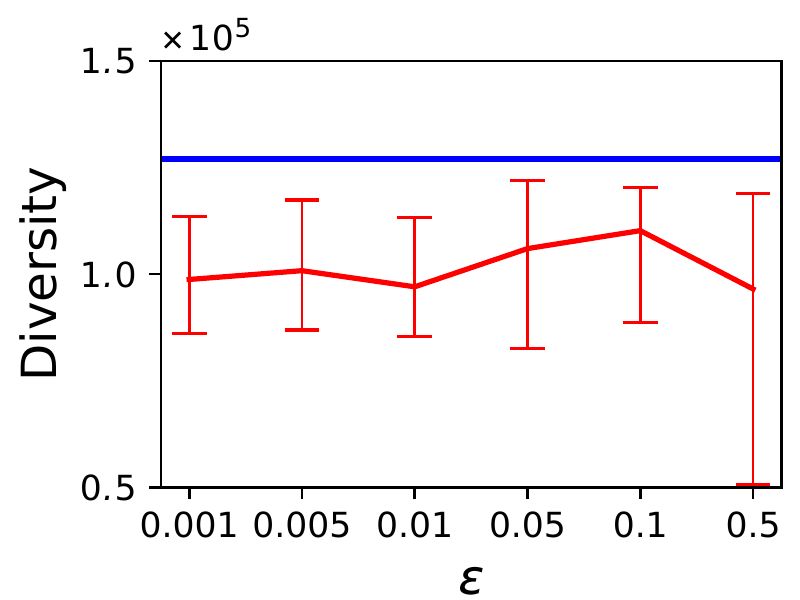}\hfill\includegraphics[width=1.1in]{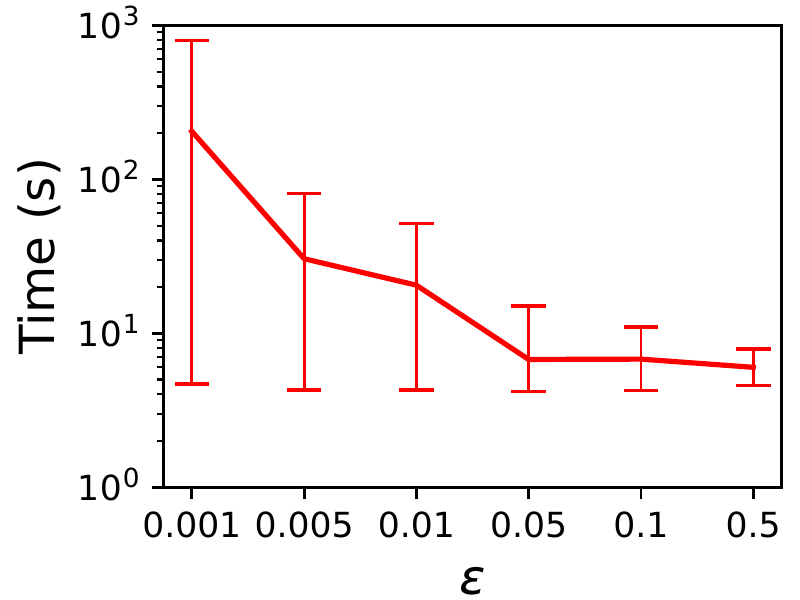}}
  \\
  \subcaptionbox{Census (Sex)}[.325\linewidth]{\includegraphics[width=1.1in]{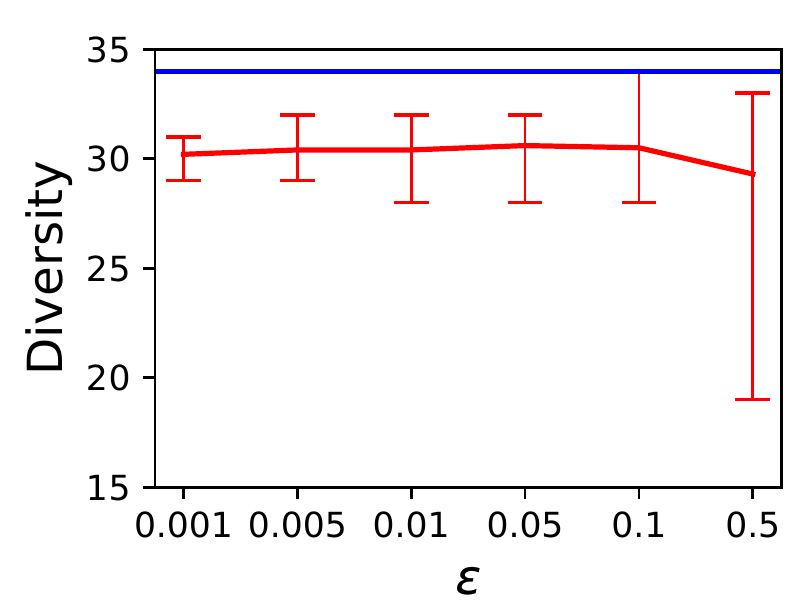}\hfill\includegraphics[width=1.1in]{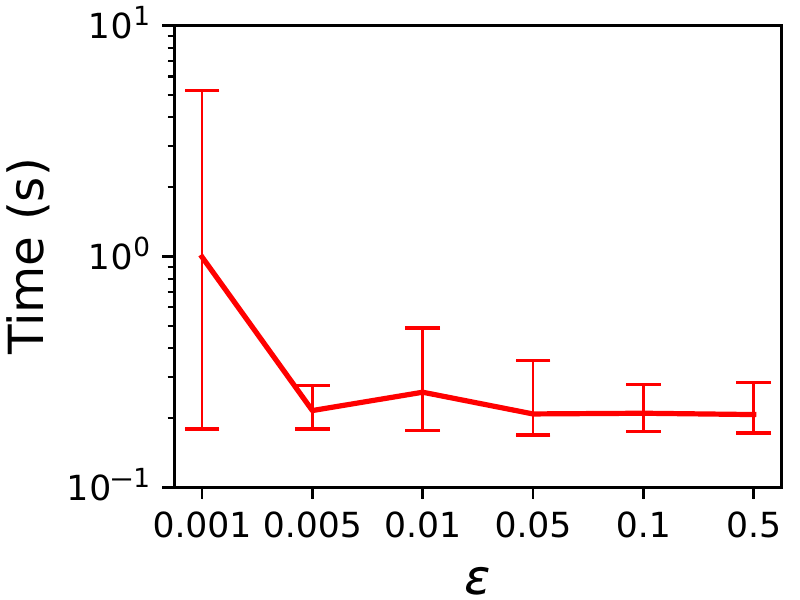}}
  \hfill
  \subcaptionbox{Census (Age)}[.325\linewidth]{\includegraphics[width=1.1in]{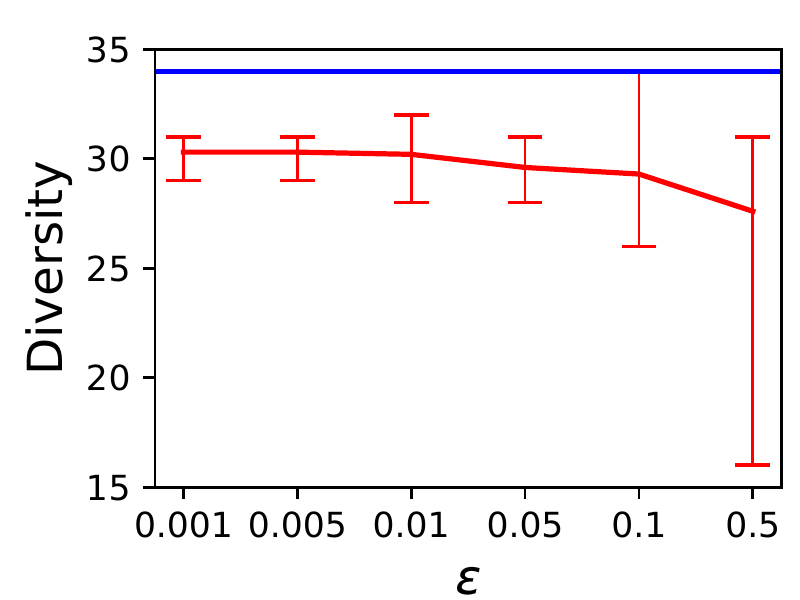}\hfill\includegraphics[width=1.1in]{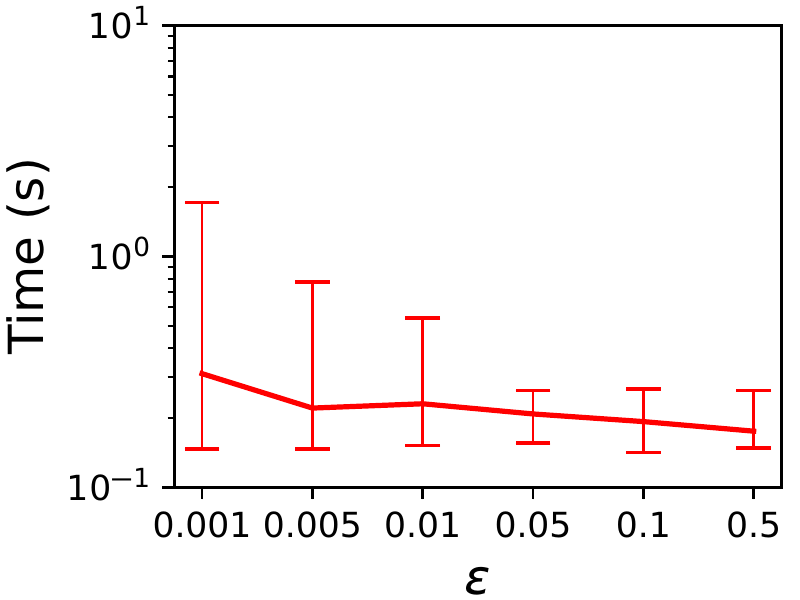}}
  \hfill
  \subcaptionbox{Twitter (Sex)}[.325\linewidth]{\includegraphics[width=1.1in]{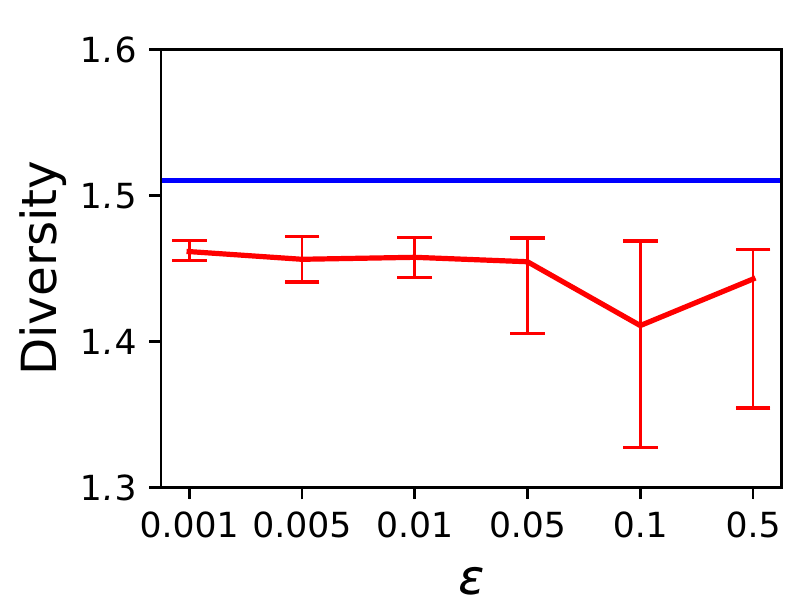}\hfill\includegraphics[width=1.1in]{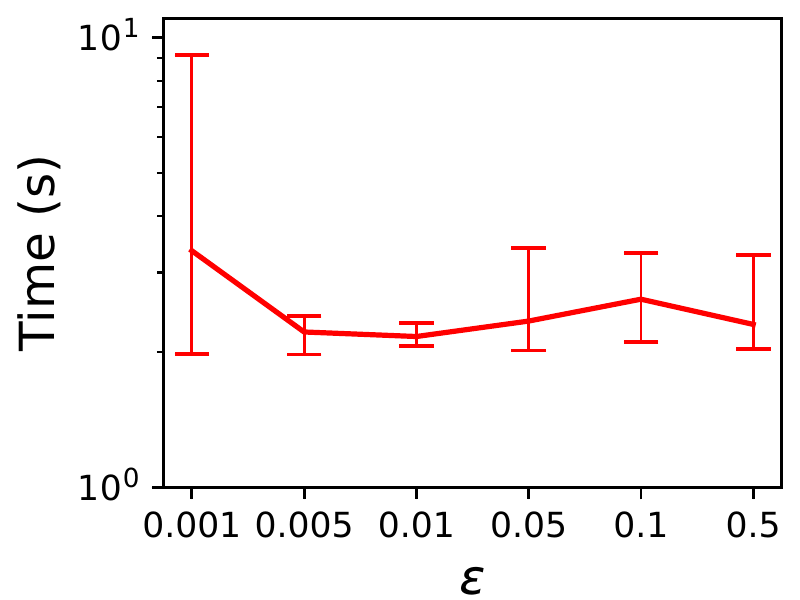}}
  \caption{Results of \textsf{FMMD-S} with varying the parameter $\varepsilon$ for solution size $k = 10$. The error bars in each plot indicate the maximum and minimum of the diversity value and running time over $10$ runs, respectively. In all plots for diversity values, the blue horizontal lines denote the diversity values of the optimal solutions provided by \textsf{FMMD-E}, which measure the gaps between the approximate solutions of \textsf{FMMD-S} and the optimal ones.}\label{fig:eps}
\end{figure}

\section{Parameter Tuning for FMMD-S}
\label{sec:add:exp}

Figure~\ref{fig:eps} illustrates the performance of \textsf{FMMD-S} by varying the parameter $\varepsilon$ from $0.001$ to $0.5$ for solution size $k=10$ on small datasets ($n = 1,000$). Regarding solution quality, the diversity values of the solutions of \textsf{FMMD-S} remain approximately constant when $\varepsilon$ takes small values. However, they decrease and become less stable as $\varepsilon$ becomes larger. Especially, when $\varepsilon = 0.5$, \textsf{FMMD-S} cannot provide stable and high-quality solutions in most cases. In terms of efficiency, the running time increases significantly and becomes less stable for smaller $\varepsilon$, particularly when $\varepsilon \leq 0.01$. These results conform to our theoretical analyses of \textsf{FMMD-S}. The above results show that the most appropriate value of $\varepsilon$ lies in the range $[0.01, 0.1]$ in most cases. Therefore, we decide to set $\varepsilon = 0.05$ in all other experiments throughout this paper.

\end{document}